\documentclass[format=sigplan,screen]{acmart}
\usepackage{flushend}
\usepackage{enumitem}
\excludecomment{techrpt}
\includecomment{paper}

\newcommand{\techrptonly}[1]{\ignorespacesafterend}
\newcommand{\paperonly}[1]{#1}
\newcommand{\material}{online technical appendix~\cite{bowman2018:cccc:tr}}

\usepackage[paper,omit,magicref]{mttex/mttex}
\usepackage{tikz-cd}
\definecolor{acmartlink}{RGB}{134,0,106}
\hypersetup{
  linkcolor = acmartlink
}

\LetLtxMacro{\oldtt}{\tt}
\input{defs}
\renewcommand{\judgshape}[2][]{\begin{flushleft}\fbox{\inlinemath{#2}}~{\small #1}\end{flushleft}}

\LetLtxMacro{\oldmathpar}{\mathpar}
\LetLtxMacro{\oldendmathpar}{\endmathpar}

\LetLtxMacro{\olddisplaymath}{\displaymath}
\LetLtxMacro{\oldenddisplaymath}{\enddisplaymath}

\LetLtxMacro{\oldalign}{\align}
\LetLtxMacro{\oldendalign}{\endalign}

\LetLtxMacro{\oldfigure}{\figure}
\LetLtxMacro{\oldendfigure}{\endfigure}

\LetLtxMacro{\oldcaption}{\caption}
\LetLtxMacro{\oldparagraph}{\paragraph}
\LetLtxMacro{\oldsection}{\section}
\LetLtxMacro{\oldsubsection}{\subsection}

\renewenvironment{mathpar}{
\begin{nop}\small\oldmathpar}{
\oldendmathpar\end{nop}\ignorespacesafterend}

\renewenvironment{displaymath}{
\begin{nop}\small\olddisplaymath}{
\oldenddisplaymath\end{nop}\ignorespacesafterend}

\renewenvironment{align}{
\small\oldalign}{
\oldendalign\ignorespacesafterend}

\renewenvironment{figure}
{\oldfigure}
{\vspace{-3ex}\oldendfigure}

\renewcommand{\inlinemath}[1]{{\small\(#1\)}}
\renewcommand{\caption}[1]{\vspace{-1.5ex}\vspace{-\baselineskip}\oldcaption{#1}}

\LetLtxMacro{\oldtfonttext}{\tfonttext}
\LetLtxMacro{\oldsfonttext}{\sfonttext}
\LetLtxMacro{\oldbefonttext}{\befonttext}

\renewcommand{\tfonttext}[1]{\oldtfonttext{\small #1}}
\renewcommand{\sfonttext}[1]{\oldsfonttext{\small #1}}
\renewcommand{\befonttext}[1]{\oldbefonttext{\small #1}}

\LetLtxMacro{\oldrulename}{\rulename}

\renewcommand{\rulename}[1]{\oldrulename{\small #1}}

\renewcommand{\paragraph}{\oldparagraph}

\author{William J. Bowman}
\affiliation{
  \institution{Northeastern University}
  \country{USA}
}
\affiliation{
  \institution{Inria Paris}
  \country{France}
}
\email{wjb@williamjbowman.com}

\author{Amal Ahmed}
\affiliation{
  \institution{Northeastern University}
  \country{USA}
}
\affiliation{
  \institution{Inria Paris}
  \country{France}
}
\email{amal@ccs.neu.edu}

\title{Typed Closure Conversion for the Calculus of Constructions}
 \titlenote{We use a combination of colors and fonts to distinguish different languages. Although the
  languages are distinguishable in black-and-white, the paper is easier to read when viewed or printed
  in color.}

\setcopyright{acmlicensed}
\acmPrice{15.00}
\acmDOI{10.1145/3192366.3192372}
\acmYear{2018}
\copyrightyear{2018}
\acmISBN{978-1-4503-5698-5/18/06}
\acmConference[PLDI'18]{39th ACM SIGPLAN Conference on Programming Language Design and Implementation}{June 18--22, 2018}{Philadelphia, PA, USA}



\begin{document}
\begin{abstract}
Dependently typed languages such as Coq are used to specify and verify the full functional
correctness of source programs.
Type-preserving compilation can be used to preserve these specifications and proofs of correctness
through compilation into the generated target-language programs.
Unfortunately, type-preserving compilation of dependent types is hard.
In essence, the problem is that dependent type systems are designed around high-level compositional abstractions
to decide type checking, but compilation interferes with the type-system rules for reasoning about
run-time terms.

We develop a type-preserving closure-conversion translation from the Calculus of Constructions
(CC) with strong dependent pairs ($\Sigma$ types)---a subset of the core language of Coq---to
a type-safe, dependently typed compiler intermediate language named CC-CC.
The central challenge in this work is how to translate the source type-system rules for reasoning about
functions into target type-system rules for reasoning about closures.
To justify these rules, we prove soundness of CC-CC by giving a model in CC.
In addition to type \nonbreaking{preservation, we prove correctness of separate compilation.}
\end{abstract}
\maketitle

\section{Introduction}
\label{sec:intro}
Full-spectrum dependently typed programming languages such as Coq have had tremendous impact on the
formal verification of large-scale software.
Coq has been used to specify and prove the full functional correctness the CompCert C
compiler~\cite{leroy2009:compcert-jfp}, the CertiKOS OS kernel~\cite{gu2016,gu2015}, and implementations of
cryptographic primitives and protocols~\cite{barthe2009,appel2015}.
The problem is that these proofs are about \emph{source programs}, but we need guarantees about the
\emph{target programs}, generated by compilers, that actually end up running on machines.
Projects such as CertiCoq~\cite{anand2017}, which aims to build a verified compiler for
Coq in Coq, are a good first step.
Unfortunately, CertiCoq throws out type information before compilation.
This makes it difficult to ensure that the invariants of verified programs are respected when linking.
A similar problem occurs when we extract a proven correct Coq program
\im{e} to OCaml, then link with some unverified OCaml component \im{f}
that violates the invariants of \im{e} and causes a segfault.
Since Coq types are not preserved into OCaml, there is no way to type check
\im{f} and flag that we should not link \im{f} with \im{e}.
The state of the art is to tell the programmer to be careful.

% Lesson of PLDI 14: spell it out, man.
Type-preserving compilation is the key to solving this problem.
Types are useful for enforcing invariants in source
programs, and we can similarly use them to check invariants when linking target programs.
With type-preserving compilation, we could compile \im{e} and preserve its specifications into
a typed target language.
Then we could use type checking at link time to verify that all components match the invariants that
\im{e} was originally verified against.
Once we have a whole program after linking all components in a low-level typed---perhaps dependently typed---assembly language,
there would no longer be a need to enforce invariants, so types could be erased to generate (untyped)
machine code.
Preserving \emph{full-spectrum} dependent types has additional benefits---we could
preserve proofs of full functional correctness into the generated code!

The goal in type-preserving compilation is not to develop new compiler translations, but to adapt
existing translations so that they perform the same function but also preserve typing invariants.
Unfortunately, these two goals are in conflict, particularly as the typing invariants become richer.
The richer the invariants the type system can express, the less freedom the compiler is permitted, and
the more work required to establish typing invariants in the transformed code.

In the case of full-spectrum dependently typed languages, type-preserving compilation is hard.
%For instance, \citet{barthe2002} show that the standard call-by-name continuation-passing style (CPS)
%transformation is \emph{not type preserving} for languages with strong dependent pairs (\(\Sigma\)
%types), and go on to show that it is \emph{impossible} to develop a type-preserving CPS translation
%for languages with dependent case analysis on sum types if the CPS translation can be used to
%implement control effects.
%(Both \(\Sigma\) types and dependent case analysis are characteristic features of full-spectrum
%dependently typed languages.)
The essential problem is that compiler transformations disrupt the syntactic reasoning used by the
type system to decide type checking.
With full-spectrum dependent types, any runtime term can appear in types, so the type system includes
rules for reasoning about equivalence and sometimes partially evaluating runtime terms during type
checking.
This works well in high level, functional languages such as the core language of Coq, but when
compilers transform high-level language concepts into low-level machine concepts, we need new rules
for how to reason about terms during type checking.

\newcommand{\mcloe}[2]{\langle\!\langle#1,#2\rangle\!\rangle}
In the case of closure conversion, the problem is that, unlike in simply typed languages, free term
variables are bound in \emph{types} as well as terms.
Intuitively, we translate a simply typed function \im{\Gamma \vdash \lambda x:A.e : A \to
B} into a closure \im{\Gamma \vdash \mcloe{(\lambda \Gamma,x:A.e)}{\dom{\Gamma}} : A \to B} where the code of the function
is paired with its environment, and the code now receives its environment as an explicit argument.
Note that the environment is \emph{hidden} in the type of the closure so that two functions of the
same type but with different environment still have the same type.\footnote{Normally, we use existential
types to hide the environment, but as we will see in \fullref[]{sec:idea}, existential types cause
problems with dependent types.}
With dependent types, the \emph{type} of a closure may refer to free variables from the
environment.
That is, in \im{\Gamma \vdash \lambda x:A.e : \Pi x:A.B}, variables from \im{\Gamma} can appear in
\im{A} and \im{B}.
After closure conversion, how can we keep the environment hidden in the type when the type must refer
to the environment?
That is, in the closure converted version of the above example \im{\Gamma \vdash \mcloe{(\lambda
\Gamma,x:A.e)}{env} : \Pi x:A.B}, how can \im{A} and \im{B} refer to \im{env} if \im{env} must remain
hidden in the type?

We solve this problem for type-preserving closure conversion of the Calculus of Constructions with
\(\Sigma\) types (CC)---a subset of the core language of Coq, and a calculus that is representative
of full-spectrum dependently typed languages.
Closure conversion transforms first-class functions with free variables into closures that pair
closed, statically allocated code with a dynamically allocated environment containing the values
of the free variables.
There are two major challenges in designing new type-system rules for closures, which we discuss at a
high-level in \fullref[]{sec:idea} before we formally present our results.
In short, we need new type-system rules for reasoning about closures, and a way to synchronize the
\emph{type} of a closure, which depends on free variables, with the type of (closed) code, which
cannot depend on free variables.

\paragraph{Contributions} We make the following contributions:
\begin{enumerate}[leftmargin=*]
  \item We design and prove the consistency of CC-CC, a full-spectrum dependently typed compiler IL
    with support for statically reasoning about closures, \fullref[]{sec:target}.
    The proof of consistency also guarantees type safety of any programs in CC-CC---\ie, linking
    any two components in CC-CC is guaranteed to have well-defined behavior.
  \item We give a typed closure-conversion translation from CC to CC-CC
    \fullref[]{sec:cc}.
  \item Leveraging the type-preservation proof, we prove that this translation is correct with respect
    to separate compilation, \ie, linking components in CC and then running to a value is equivalent
    to first compiling the components separately and then linking in CC-CC.
    \omitthis{
    \item From the type-preservation proof and our proof of consistency for CC-CC, we prove
    preservation and reflection of definitional equivalence \fullref[]{sec:cc}.
    This theorem resembles \emph{full abstraction}, which has applications in secure
    compilation, although full abstraction requires preservation and reflection of contextual equivalence.
  }
\end{enumerate}

Next, we introduce CC (\fullref[]{sec:source}), both to introduce our source language and to formally
introduce dependent types, before presenting the central problem with typed closure conversion, and the
main idea behind our solution (\fullref[]{sec:idea}).
Elided parts of figures and proofs are included in our \material.
%%% Local Variables:
%%% mode: latex
%%% TeX-master: "paper"
%%% End:

\newcommand{\FigCCSyntax}[1][t]{
  \begin{figure}[#1]
    \begin{displaymath}
      \begin{array}{lr@{\hspace{1em}}c@{\hspace{1em}}l}
        \bnflabel{Universes} &
        \sU & \!\!\bnfdef & \sstarty \bnfalt \sboxty
        \techrptonly{\\[6pt]}
        \paperonly{\\}

        \bnflabel{Expressions} & \se,\sA,\sB & \bnfdef & \sx
                                                         \bnfalt \sstarty
                                                         \bnfalt \salete{\sx}{\se}{\sA}{\se}
                                                         \bnfalt \spity{\sx}{\sA}{\sB}
        \paperonly{\\ &&}                    \bnfalt \paperonly{&}
                                                         \sfune{\sx}{\sA}{\se}
                                                         \bnfalt \sappe{\se}{\se}
                                                         \bnfalt \ssigmaty{\sx}{\sA}{\sB}
        \\
                             &&\bnfalt& \sdpaire{\seone}{\setwo}{\ssigmaty{\sx}{\sA}{\sB}}
                                        \bnfalt \sfste{\se}
                                        \bnfalt \ssnde{\se}
        \techrptonly{\\[6pt]}
        \paperonly{\\}

        \bnflabel{Environments} &
        \slenv & \bnfdef & \cdot \bnfalt \slenv,\sx:\sA \bnfalt \slenv,\sx = \se:\sA
        \techrptonly{\\[6pt]}
        \paperonly{\\}
      \end{array}
    \end{displaymath}
    \paperonly{\vspace{1ex}}
    \caption{CC Syntax}
    \label{fig:cc:syntax}
    \paperonly{\vspace{-1ex}}
  \end{figure}
}

\newcommand{\FigCCConv}[1][t]{
  \begin{figure}[#1]
    \judgshape{\sstepjudg{\slenv}{\se}{\sepr}}
    \paperonly{\vspace{-1.5ex}}
    \begin{displaymath}
      \begin{array}{rlll}
        \sx & \step_{\delta} & \se & \where{\sx = \se : \sA \in \slenv}
        \techrptonly{\\[4pt]}
        \paperonly{\\}
        \salete{\sx}{\se}{\sA}{\seone} & \step_{\zeta} & \subst{\seone}{\se}{\sx}
        \techrptonly{\\[4pt]}
        \paperonly{\\}
        \sappe{(\sfune{\sx}{\sA}{\seone})}{\setwo} & \step_{\beta} & \subst{\seone}{\setwo}{\sx}
        \techrptonly{\\[4pt]}
        \paperonly{\\}
        \sfste{\spaire{\seone}{\setwo}} & \step_{\pi_{1}} & \seone

        \techrptonly{\\[4pt]}
        \paperonly{\\}
        \ssnde{\spaire{\seone}{\setwo}} & \step_{\pi_{2}} & \setwo
      \end{array}
    \end{displaymath}
    \judgshape{\sequivjudg{\slenv}{\se}{\sepr}}
    \paperonly{\vspace{-1ex}}
    \begin{mathpar}
      \inferrule*[right=\rulename{\figm{\equiv}}]
        {\sstepjudg[\stepstar]{\slenv}{\seone}{\se} \\
         \sstepjudg[\stepstar]{\slenv}{\setwo}{\se}}
        {\sequivjudg{\slenv}{\seone}{\setwo}}

      \inferrule*[right=\rulename{\figm{\equiv}-\figm{\eta_1}}]
        {\sstepjudg[\stepstar]{\slenv}{\seone}{\sfune{\sx}{\sA}{\se}} \\
         \sstepjudg[\stepstar]{\slenv}{\setwo}{\setwopr} \\
         \sequivjudg{\slenv,\sx:\sA}{\se}{\sappe{\setwopr}{\sx}}}
        {\sequivjudg{\slenv}{\seone}{\setwo}}

      \inferrule*[right=\rulename{\figm{\equiv}-\figm{\eta_2}}]
        {\sstepjudg[\stepstar]{\slenv}{\seone}{\seonepr} \\
         \sstepjudg[\stepstar]{\slenv}{\setwo}{\sfune{\sx}{\sA}{\se}} \\
         \sequivjudg{\slenv,\sx:\sA}{\sappe{\seonepr}{\sx}}{\se}}
        {\sequivjudg{\slenv}{\seone}{\setwo}}
    \end{mathpar}
    \caption{CC Conversion and Equivalence}
    \label{fig:cc:conv}
    \paperonly{\vspace{-1ex}}
  \end{figure}
}

\newcommand{\FigCCTyping}[1][t]{
  \begin{figure}[#1]
    \judgshape{\styjudg{\slenv}{\se}{\sA}}
    \paperonly{\vspace{-.5ex}}
    \begin{mathpar}
      \inferrule*[right=\rulename{Ax-*}]
      {\swf{\slenv}}
      {\styjudg{\slenv}{\sstarty}{\sboxty}}

      \inferrule*[right=\rulename{Var}]
      {(\sx : \sA \in \slenv \text{ or } \sx = \se : \sA \in \slenv) \\
        \swf{\slenv}}
      {\styjudg{\slenv}{\sx}{\sA}}

      \inferrule*[right=\rulename{Let}]
      {\styjudg{\slenv}{\se}{\sA} \\
        \styjudg{\slenv,\sx=\se:\sA}{\sepr}{\sB}}
      {\styjudg{\slenv}{\salete{\sx}{\se}{\sA}{\sepr}}{\subst{\sB}{\se}{\sx}}}

     \inferrule*[right=\rulename{Prod-*}]
        {\styjudg{\slenv,\sx:\sA}{\sB}{\sstarty}}
        {\styjudg{\slenv}{\spity{\sx}{\sA}{\sB}}{\sstarty}}

     \inferrule*[right=\rulename{Prod-\(\square\)}]
        {\styjudg{\slenv,\sx:\sA}{\sB}{\sboxty}}
        {\styjudg{\slenv}{\spity{\sx}{\sA}{\sB}}{\sboxty}}

      \inferrule*[right=\rulename{Lam}]
        {\styjudg{\slenv,\sx:\sA}{\se}{\sB}}
        {\styjudg{\slenv}{\sfune{\sx}{\sA}{\se}}{\spity{\sx}{\sA}{\sB}}}

        \inferrule*[right=\rulename{App}]
        {\styjudg{\slenv}{\se}{\spity{\sx}{\sApr}{\sB}} \\
         \styjudg{\slenv}{\sepr}{\sApr}}
        {\styjudg{\slenv}{\sappe{\se}{\sepr}}{\subst{\sB}{\sepr}{\sx}}}

        \inferrule*[right=\rulename{Sig-*}]
        {\styjudg{\slenv}{\sA}{\sstarty} \\
          \styjudg{\slenv,\sx:\sA}{\sB}{\sstarty}}
        {\styjudg{\slenv}{\ssigmaty{\sx}{\sA}{\sB}}{\sstarty}}

        \inferrule*[right=\rulename{Sig-\(\square\)}]
        {\styjudg{\slenv,\sx:\sA}{\sB}{\sboxty}}
        {\styjudg{\slenv}{\ssigmaty{\sx}{\sA}{\sB}}{\sboxty}}

        \inferrule*[right=\rulename{Fst}]
        {\styjudg{\slenv}{\se}{\ssigmaty{\sx}{\sA}{\sB}}}
        {\styjudg{\slenv}{\sfste{\se}}{\sA}}

        \inferrule*[right=\rulename{Snd}]
        {\styjudg{\slenv}{\se}{\ssigmaty{\sx}{\sA}{\sB}}}
        {\styjudg{\slenv}{\ssnde{\se}}{\subst{\sB}{\sfste{\se}}{\sx}}}

      \inferrule*[right=\rulename{Conv}]
        {\styjudg{\slenv}{\se}{\sA} \\
         \styjudg{\slenv}{\sB}{\sU} \\
         \sequivjudg{\slenv}{\sA}{\sB}}
        {\styjudg{\slenv}{\se}{\sB}}
    \end{mathpar}
    \caption{CC Typing}
    \label{fig:cc:type}
  \end{figure}
}

\newcommand{\FigCCWF}[1][t]{
  \begin{figure}[#1]
    \judgshape{\swf{\slenv}}
    \paperonly{\vspace{-1ex}}
    \begin{mathpar}
        \inferrule*[right=\rulename{W-Empty}]
        {~}
        {\swf{\cdot}}

      \inferrule*[right=\rulename{W-Assum}]
        {\swf{\slenv} \\
         \styjudg{\slenv}{\sA}{\sU}}
        {\swf{\slenv,\sx:\sA}}

      \inferrule*[right=\rulename{W-Def}]
        {\swf{\slenv} \\
         \styjudg{\slenv}{\se}{\sA} \\
         \styjudg{\slenv}{\sA}{\sU}}
        {\swf{\slenv,\sx = \se :\sA}}
    \end{mathpar}
    \paperonly{\vspace{1ex}}
    \caption{CC Well-Formed Environments}
    \label{fig:cc:wf}
    \paperonly{\vspace{1ex}}
  \end{figure}
}

\newcommand{\FigCCSubst}[1][t]{
  \begin{figure}[#1]
    \begin{displaymath}
      \begin{array}{lrcl}
        \bnflabel{Closing Substitution} &
        \ssubst & \bnfdef & \cdot \bnfalt \mapext{\ssubst}{\sx}{\se}
      \end{array}
    \end{displaymath}

    \judgshape{\ssubstok{\slenv}{\ssubst}} (valid closing substitution)
    \begin{mathpar}
      \ssubstok{\slenv}{\ssubst} \defeq \forall \sx:\st \in \slenv. \styjudg{}{\ssubst(\sx)}{\st}.
    \end{mathpar}
  \caption{CC Linking}
  \label{fig:cc:link}
  \end{figure}
}

\section{Source: Calculus of Constructions (CC)}
\label{sec:source}

\FigCCSyntax
Our source language is a variant of the Calculus of Constructions (CC) extended with strong dependent
pairs (\(\Sigma\) types) and \(\eta\)-equivalence for functions, which we typeset in a
\sfonttext{non-bold, blue, sans-serif font}.
This model is based on the CIC specification used in Coq~\cite[Chapter 4]{coq2017}.
For brevity, we omit base types from this formal system but will freely use base types like natural
numbers in examples.

We present the syntax of CC in \fullref[]{fig:cc:syntax}.
Universes, or sorts, \im{\sU} are essentially the types of types.
CC includes one impredicative universe \im{\sstarty}, and one predicative universe \im{\sboxty}.
Expressions have no explicit distinction between terms, types, or kinds, but we usually use the
meta-variable \im{\se} to evoke a term expression and \im{\sA} or \im{\sB} to evoke a type expression.
Expressions include names \im{\sx}, the universe \im{\sstarty}, functions
\im{\sfune{\sx}{\sA}{\se}}, application \im{\sappe{\seone}{\setwo}}, dependent function types
\im{\spity{\sx}{\sA}{\sB}}, dependent let \im{\salete{\sx}{\se}{\sA}{\sepr}}, \(\Sigma\) types
\im{\ssigmaty{\sx}{\sA}{\sB}}, dependent pairs
\im{\sdpaire{\seone}{\setwo}{\ssigmaty{\sx}{\sA}{\sB}}}, first projections \im{\sfste{\se}} and second
projections \im{\ssnde{\se}}.
The universe \im{\sboxty} is only used by the type system and is not a valid term.
As syntactic sugar, we omit the type annotations on dependent let \im{\slete{\sx}{\se}{\sepr}} and on
dependent pairs \im{\spaire{\seone}{\setwo}} when they are irrelevant or obvious from context.
We also write function types as \im{\sfunty{\sA}{\sB}} when the result \im{\sB} does not depend on the
argument.
Environments \im{\slenv} include assumptions \im{\sx:\sA} that a name \im{\sx} has type
\im{\sA}, and definitions \im{\sx = \se : \sA} that name \im{\sx} refers to \im{\se} of type \im{\sA}.

\FigCCConv
We define conversion, or reduction, and definitional equivalence for CC in \fullref[]{fig:cc:conv}.
Conversion here is defined for deciding equivalence between types (which include terms), but it can
also be viewed as the operational semantics of CC terms.
The small-step reduction \im{\sstepjudg{\slenv}{\se}{\sepr}} reduces the expression \im{\se} to the
term \im{\sepr} under the local environment \im{\slenv}, which we usually leave implicit
for brevity.
The local environment is necessary to convert a name to its definition.
Each conversion rule is labeled, and when we refer to conversion with an unlabeled arrow \im{\se \step
\sepr}, we mean that \im{\se} reduces to \im{\sepr} by \emph{some} reduction rule, \ie, either
\im{\step_\delta}, \im{\step_\zeta}, \im{\step_\beta}, \im{\step_{\pi_1}}, or \im{\step_{\pi_2}}.
We write \im{\sstepjudg[\stepstar]{\slenv}{\se}{\sepr}} to mean the reflexive, transitive,
contextual closure of the relation \im{\sstepjudg[\step]{\slenv}{\se}{\sepr}}.
Essentially, \im{\se \stepstar \sepr} runs \im{\se} using the \im{\step} relation any number of
times, under any arbitrary context.
\techrptonly{
  The \im{\stepstar} relation introduces a definition into the local environment when descending into
  the body of a dependent let.
  That is, we have the following closure rule for \im{\stepstar}.

  \begin{mathpar}
    \inferrule
    {\sstepjudg[\stepstar]{\slenv,\sx = \se}{\seone}{\setwo}}
    {\sstepjudg[\stepstar]{\slenv}{\slete{\sx}{\se}{\seone}}{\slete{\sx}{\se}{\setwo}}}
  \end{mathpar}
}

\noindent We define equivalence \im{\sequivjudg{\slenv}{\se}{\sepr}} as reduction in the
\im{\stepstar} relation up to \(\eta\)-equivalence, as in Coq~\cite[Chapter 4]{coq2017}.

\FigCCTyping
In \fullref[]{fig:cc:type}, we present the typing rules.
The type system is standard.

Functions \im{\sfune{\sx}{\sA}{\se}} have dependent function type \im{\spity{\sx}{\sA}{\sB}}
(\rulename{Lam}).
The dependent function type describes that the function takes an argument, \im{\sx}, of type \im{\sA},
and returns something of type \im{\sB} where \im{\sB} may refer to, \ie, \emph{depends on}, the
value of the argument \im{\sx}.
We can use this to write polymorphic functions, such as the polymorphic identity function described by
the type \im{\spity{\sA}{\sstarty}{\spity{\sx}{\sA}{\sA}}}, or functions with pre/post
conditions, such as the division function described by
\im{\spity{\sx}{\sNat}{\spity{\sy}{\sNat}{\spity{\_}{\sfont{\sy > 0}}{\sNat}}}}, which statically ensures
that we never divide by zero by requiring a proof that its second argument is greater than zero.

Applications \im{\sappe{\seone}{\setwo}} have type \im{\subst{\sB}{\setwo}{\sx}} (\rulename{App}),
\ie, the result type \im{\sB} of the function \im{\seone} with the argument \im{\setwo} substituted
for the name of the argument \im{\sx}.
Using this rule and our example of the division function \im{\sfont{div} :
  \spity{\sx}{\sNat}{\spity{\sy}{\sNat}{\spity{\_}{\sfont{\sy > 0}}{\sNat}}}}, we type check the term
\im{\sappe{\sfont{div}}{\sfont{4~2}} : {\spity{\_}{\sfont{2 > 0}}{\sNat}}}.
Notice that the term variable \im{\sy} in the type has been replaced with the value of the argument
\im{2}.

Dependent pairs \im{\spaire{\seone}{\setwo}} have type \im{\ssigmaty{\sx}{\sA}{\sB}}
(\rulename{Pair}).
Again, this type is a binding form.
The type \im{\sB} of the second component of the pair can refer to the first component of the pair by
the name \im{\sx}.
We see in the rule \rulename{Snd} that the type of \im{\ssnde{\se}} is
\im{\subst{\sB}{\sfste{\se}}{\sx}}, \ie, the type \im{\sB} of the second component of the pair with
the name \im{\sx} substituted by \im{\sfste{\se}}.
We can use this to encode refinement types, such as the describing positive numbers by
\im{\ssigmaty{\sx}{\sNat}{\sfont{\sx > 0}}}, \ie, a pair of a number \im{\sx} with a proof that
\im{\sx} is greater than \im{\sfont{0}}.

Since types are also terms, we have typing rules for types.
The type of \im{\sstarty} is \im{\sboxty}.
We call \im{\sstarty} the universe of small types and \im{\sboxty} the universe of large types.
Intuitively, small types are the types of programs while large types are the types of types and
type-level computations.
Since no user can write down \im{\sboxty}, we need not worry about the type of \im{\sboxty}.
In \rulename{Prod-*}, we assign the type \im{\sstarty} to the dependent function type when the result
type is also \im{\sstarty}.
This rule allows \emph{impredicative} functions, since it allows forming a function that quantifies
over large types but is in the universe of small types.
The rule \rulename{Prod-\im{\square}} looks similar, but is implicitly predicative, since there is no
universe larger than \im{\sboxty} to quantify over.
(We could combine the rules for \(\Pi\), but explicit separation helps clarify the issue of
predicativity when compared with the rules for \(\Sigma\) types, which cannot be combined.)
Formation rules for \(\Sigma\) types have an important restriction: it is
unsound to allow impredicativity in strong dependent pairs~\cite{coquand1986,hook1986}.
The \rulename{Sig-*} rule only allows quantifying over a small type when forming a small dependent pair.
The \rulename{Sig-\im{\square}} rule allows quantifying over either small or large types when forming a
large \(\Sigma\).
As usual in models of dependent type theory, we exclude base types, although they are simple
to add.

The rule \rulename{Conv} allows resolving type equivalence and reducing terms in types.
For instance, if we want to show that \im{\se : \ssigmaty{\sx}{\sNat}{\sfont{\sx = 2}}} but we have \im{\se :
  \ssigmaty{\sx}{\sNat}{\sfont{\sx = 1 + 1}}}, the \rulename{Conv} rule performs this reduction.
Note while our equivalence relation is untyped, the \rulename{Conv}
rule ensures that \im{\sA} and \im{\sB}
are well-typed before appealing to equivalence, ensuring decidability.
\mbox{(It is a standard lemma that if \im{\styjudg{\slenv}{\se}{\sA}}, then \im{\styjudg{\slenv}{\sA}{\sU}}~\cite{luo1989}.)}

Finally, we extend well-typedness to well-formedness of environments \im{\swf{\slenv}} in
\fullref[]{fig:cc:wf}.

\FigCCWF
%%% Local Variables:
%%% mode: latex
%%% TeX-master: "paper"
%%% End:

\section{Main Ideas}
\label{sec:idea}

Closure conversion makes the implicit closures from a functional language explicit to facilitate
statically allocating functions in memory.
The idea is to translate each first-class function into an explicit closure, \ie, a pair of closed
\emph{code} and an environment data structure containing the values of the free variables.
We use \emph{code} to refer to functions with no free variables, as in a closure-converted language.
The environment is created dynamically, but the closed code can be lifted to the top-level and
statically allocated.
Consider the following example translation.

\begin{inlinedisplay}
  \begin{array}{@{\hspace{0pt}}r@{\hspace{1ex}}c@{\hspace{1ex}}l}
  (\lambda x.y)^+ &=& \left<(\lambda
                      n\,x.\,\kwopen{\text{let}}y\mathrel{=}(\kwopen{\pi_1}n)\mathrel{\text{in}}
                      y), \left<y\right>\right> \\
  ((\lambda x.y)~\text{true})^+ &=& \begin{stackTL}
                           \kwopen{\text{let}}\left<f, n\right>\mathrel{=}\left<(\lambda
                               n\,x.\,\kwopen{\text{let}}y\mathrel{=}(\kwopen{\pi_1}n)\mathrel{\text{in}}y),
                            \left<y\right>\right>\mathrel{\text{in}}
                           \\\quad f~n~\text{true}
                          \end{stackTL}
    \end{array}
\end{inlinedisplay}

We write \im{e^+} to indicate the translation of an expression \im{e}.
We translate each function into a pair of code and its environment.
The code accepts its free variables in an environment argument, \im{n} (since \im{n} sounds similar to
\emph{env}).
In the body of the code, we bind the names of all free variables by projecting from this
environment \im{n}.
To call a closure, we apply the code to its environment and its argument.

This translation is not type preserving since the structure of the environment shows up in the type.
For example, the following two functions have the same type in the source, but end up with different
types in the target.

\begin{inlinedisplay}
  \begin{array}{rcl}
  (\lambda x.y)^+ &:&  ((\text{Nat} \times \text{Nil}) \to \text{Nat} \to \text{Nat})  \times (\text{Nat} \times \text{Nil})
  \\
  (\lambda x.x)^+ &:&  (\text{Nil} \to \text{Nat} \to \text{Nat})  \times \text{Nil}
  \end{array}
\end{inlinedisplay}

This is a well-known problem with typed closure conversion, so we could try the well-known
solution~[\citealp{minamide1996}, \citealp{morrisett1998:reccc}, \citealp{morrisett1998:ftotal}, \citealp{ahmed2008},
\citealp{perconti2014}, \citealp{new2016}].
(Spoiler alert: it won't work for CC.)
We represent closures as an existential package of a pair of the function and its environment, whose
type is hidden.
The existential type hides the structure of the environment in the type.

\begin{inlinedisplay}
  \begin{array}{rcl}
    (\lambda x.y)^+ &:& \exists \alpha. (\alpha \to \text{Nat} \to \text{Nat})  \times \alpha
    \\
    (\lambda x.x)^+ &:& \exists \alpha. (\alpha \to \text{Nat} \to \text{Nat})  \times \alpha
  \end{array}
\end{inlinedisplay}

This works well for simply typed and polymorphic languages, but when we move to a dependently typed
language, we have new challenges.
First, the environment must now be ordered since the type of each new variable can depend on all prior
variables.
Second, types can now refer to variables in the closure's environment.
Recall the polymorphic identity function from earlier.

\begin{inlinedisplay}
  \begin{array}{rcl}
  \sfune{\sA}{\sstarty}{\sfune{\sx}{\sA}{\sx}} & : & \spity{\sA}{\sstarty}{\spity{\sx}{\sA}{\sA}}
  \end{array}
\end{inlinedisplay}

This function takes a type variable, \im{\sA}, whose type is \im{\sstarty}.
It returns a function that accepts an argument \im{\sx} of type \im{\sA} and returns it.
There are two closures in this example: the outer closure has no free variables, and thus will have an
empty environment, while the inner closure \im{\sfune{\sx}{\sA}{\sx}} has \im{\sA} free, and thus
\im{\sA} will appear in its environment.

Below, we present the translation of this example using the previous translation.
We typeset target language terms produced by our translation in a
\tfonttext{bold, red, serif font}.
We produce two closures, one nested in the other.
Note that we translate source variables \im{\sx} to \im{\tx}.
In the outer closure, the environment is empty \im{\tnpaire{}}, and the code simply returns the
inner closure.
The inner closure has the argument \im{\tA} from the outer code in its environment.
Since the inner code takes an argument of type \im{\tA}, we project \im{\tA} from the
environment \emph{in the type annotation} for \im{\tx}.
That is, the inner code takes an environment \im{\tntwo} that contains \im{\tA}, and the type
annotation for \im{\tx} is \im{\tx : \tfste{\tntwo}}.
The type \im{\tfste{\tntwo}} is unusual, but is no problem since dependent types allow computations in
types.

\begin{inlinedisplay}
  \begin{stackTL}
    \tcloe{\tnfune{(\tnone:\tunitty,\tA:\tstarty)}{
        \tcloe{\tnfune{(\tntwo:
            \tpairty{\tstarty}{\tunitty},\tx:\tfste{\tntwo})}{\tx}
        }{\tnpaire{\tA,\tnpaire{}}}}}{\tnpaire{}} ~~ : ~~\\
    ~~\begin{array}{@{\hspace{0pt}}l@{\hspace{0pt}}l@{\hspace{0pt}}l}
      \texistty{\talphaone}{\tstarty}{&\,\tpairty{(\tnpity{&(\tnone:\talphaone,\tA:\tstarty)}{
            \\ &&
            \texistty{\talphatwo}{\tboxty}{\tpairty{(\tnpity{(\tntwo:\talphatwo,\tx:\tfste{\tntwo})}{\tfste{\tntwo}})}{\talphatwo}}})\,}{\talphaone}}
    \end{array}
    \end{stackTL}
\end{inlinedisplay}

We see that the inner code on its own is well typed with the closed type
\im{\tnpity{(\tntwo:\tpairty{\tstarty}{\tunitty},\tx:\tfste{\tntwo})}{\tfste{\tntwo}}}.
That is, the code takes two arguments: the first argument \im{\tntwo} is the environment, and the
second argument \im{\tx} is a value of type \im{\tfste{\tntwo}}.
The result type of the code is also \im{\tfste{\tntwo}}.
As discussed above, we must hide the type of the environment to ensure type preservation.
That is, when we build the closure
\im{\tcloe{\tnfune{(\tntwo:
      \tpairty{\tstarty}{\tunitty},\tx:\tfste{\tntwo})}{\tx}
    }{\tnpaire{\tA,\tnpaire{}}}}, we must hide the type of the environment \im{\tnpaire{\tA,\tnpaire{}}}.
We use an existential type to quantify over the type \im{\talphatwo} of the environment, and we
produce the type \im{\tnpity{(\tntwo:\talphatwo,\tx:\tfste{\tntwo})}{\tfste{\tntwo}}} for the
code in the inner closure.
But this type is trying to take the first projection of something of type \im{\talphatwo}.
We can only project from pairs, and something of type \im{\talphatwo} isn't a pair!
In hiding the type of the environment to recover type preservation, we've broken type preservation for
dependent types.

A similar problem also arises when closure converting System F, since System F also features type
variables~\cite{minamide1996,morrisett1998:ftotal}.
To understand our solution, it is important to understand why the solutions that have historically
worked for System F do not scale to CC.
We briefly present these past results and why they do not scale before moving on to the key idea behind
our translation.
Essentially, past work using existential types relies on assumptions about computational relevance,
parametricity, and impredicativity that do not necessarily hold in full-spectrum dependent type
systems.

\subsection{Why the well known solution doesn't work}
\citet{minamide1996} give a translation that encodes
closure types using existential types, a standard type-theoretic
feature that they use to make environment hiding explicit in the
types. In essence, they encode closures as objects; the environment
can be thought of as the private field of an object.  Since then, the use of
existential types to encode closure types has been standard in all
work on typed closure conversion.

However, the use of existential types to encode closures in a dependently
typed setting is problematic.  First, let us just consider closure
conversion for System F. As \citet{minamide1996} observed, there is a
problem when code must be closed with respect to both term and
\emph{type} variables.  This problem is similar to the one
discussed above:
when closure environments contain type variables,
since those type variables can also appear in the closure's type, the closure's type needs
to project from the closure's (hidden) environment which has type \im{\alpha}.
To fix the problem, they extend their target language with \emph{translucency} (essentially, a kind of
type-level equivalence that we now call singleton types), type-level pairs, and kinds.
All of these features can be encoded in CC, so we could extend their translation essentially as
follows.

\begin{inlinedisplay}
  \begin{array}{rcl}
    \!\!\!(\spity{\sx}{\sA}{\sB})^+ \defeq
    \texistty{\talpha}{\tU}{\texistty{\tn}{\talpha}{\tcodety{\tnpr{\,:\,}\talpha, \ty{\,:\,}\tnpr = \tn,\tx{\,:\,}\sA^+}{\sB^+}}}
  \end{array}
\end{inlinedisplay}

In this translation, we would existentially quantify over the type of the environment \im{\talpha},
the \emph{value} of the environment \im{\tn}, and generate code that requires an environment
\im{\tnpr} plus a proof that the code is only ever given the environment \im{\tn} as the argument
\im{\tnpr}.
The typing rule for an existential package copies the existential value into the type.
That is, for a closure \im{\tnpackoe{\tApr,\tv,\te}} of type
\im{\texistty{\talpha}{\tU}{\texistty{\tn}{\talpha}{\tcodety{\tnpr{\,:\,}\talpha,
        \ty{\,:\,}\tnpr = \tn,\tx{\,:\,}\sA^+}{\sB^+}}}}, the typing rule for
\im{\tfont{pack}} requires that we show \im{\te : \tcodety{\tnpr{\,:\,}\tApr,
        \ty{\,:\,}\tnpr = \tv,\tx{\,:\,}\sA^+}{\sB^+}}; notice that the variable \im{\tn}
has been replaced by the value of the environment \im{\tv}.
The equality \im{\tnpr = \tv} essentially unifies projections from
\im{\tnpr} with projections from \im{\tv}, the list of
free variables representing the actual environment.

The problem with this translation is that it relies on \emph{impredicativity}.
That is, if \im{(\spity{\sx}{\sA}{\sB}) : \sstarty}, then we require that
\im{(\spity{\sx}{\sA}{\sB})^+ : \tstarty}.
Since the existential type quantifies over a type in an arbitrary universe \im{\tU} but must be in the
base universe \im{\tstarty}, the existential type must be impredicative.
Impredicative existential types (weak dependent sums) are consistent on their own, but impredicativity
causes inconsistency when combined with other features, including
computational relevance and Coq's universe hierarchy.
In Coq by default, the base computationally relevant universe \texttt{Set} is predicative, so this
translation would not work.
There is a flag to enable impredicative \texttt{Set}, but this can introduce inconsistency with some
axioms, such as a combination of the law of excluded middle plus the axiom of
choice, or ad-hoc polymorphism~\cite{boulier2017}.
Even with impredicative \texttt{Set}, there are computationally relevant universes higher in Coq's universe
hierarchy, and it would not be safe to allow impredicativity at more than one universe.
Furthermore, some dependently typed languages, such as Agda, do not allow impredicativity at all since
it is the source of paradoxes, such as Girard's paradox.

A second problem arises in developing an \(\eta\) principle, because
the existential type encoding relies on \emph{parametricity} to hide
the environment.
So, any \(\eta\) principle would need
to be justified by a parametric relation on environments.
Internalizing parametricity for dependent type theory is an active area of
research~\cite{krishnaswami2013,bernardy2012,keller2012,nuyts2017} and not all dependent type
theories admit parametricity~\cite{boulier2017}.

Later, \citet{morrisett1998:ftotal} improved the existential-type translation
for System F, avoiding translucency and kinds by relying on \emph{type
  erasure} before runtime, which meant that their
code didn't have to close over type variables.  This translation does
not apply in a dependently typed setting, since now types can contain
term variables not just ``type erasable'' type variables.

\subsection{Our translation}
To solve type-preserving closure conversion for CC, we avoid
existential types altogether and instead take inspiration from the
so-called ``abstract closure conversion'' of \citet{minamide1996}.
They add new forms to the target language to represent code and
closures for a simply typed source language.  We scale the design of
these forms to dependent types.

Adapting and scaling even a well-known translation to dependent type
theory is complex.
Recall from \fullref[]{sec:intro} that the goal of our compiler is to implement the same functionality
as standard closure conversion, but preserve the typing invariants.
Operationally, our translation will do the obvious thing, but the
complexity of our translation comes from the types.
In the case of dependent types, the complexity (and usefulness) of the type system comes from the
ability to interpret terms as logical formulas that are capable of expressing mathematical theorems and
proofs.
When we add new typing rules to the target language, we must justify
that the new system is still consistent when interpreted as a logic.
Moreover, we must design new equivalence rules for terms and, ideally, ensure that equivalence is still
decidable.

In the case of closure conversion, we are transforming the fundamental feature of dependent type
theory: functions and \(\Pi\) types.
Functions can be interpreted as proofs of universal properties represented by \(\Pi\) types.
This transformation requires dependent types for both code and closures, and a novel equivalence
principle for closures.
But in proving the new rules consistent, we must not just prove that
we do not allow proofs of False in the new system, but also establish
that all universal properties and their proofs that were representable
and provable in the source language are still representable and
provable in the target language.
We leave the proofs of these properties until \fullref[]{sec:consistency}, but present the key typing
and equivalence rules now.

We extend our type system with primitive types for code and closures.
We represent code as \im{\tnfune{(\tn:\tApr,\tx:\tA)}{\teone}} of
the \emph{code type} \im{\tcodety{\tn:\tApr,\tx:\tA}{\tB}}.
These are still dependent types, so \im{\tn} may appear in both \im{\tA} and \im{\tB}, and
\im{\tx} may appear in \im{\tB}.
Code must be well typed in an empty environment, \ie, when it is closed.
For simplicity, code only takes two arguments.
\paperonly{\vspace{-1ex}}

\begin{mathpar}
  \inferrule*[right=\rulename{Code}]
  {\ttyjudg{\cdot,\tn:\tApr,\tx:\tA}{\te}{\tB}}
  {\ttyjudg{\tlenv}{\tnfune{\tn:\tApr,\tx:\tA}{\te}}{\tcodety{\tn:\tApr,\tx:\tA}{\tB}}}
\end{mathpar}

We represent closures as \im{\tcloe{\te}{\tepr}} of type
\im{\tpity{\tx}{\subst{\tA}{\tepr}{\tn}}{\subst{\tB}{\tepr}{\tn}}}, where
\im{\te} is code and \im{\tepr} is its environment.
We continue to use \im{\tfontsym{\Pi}} types to describe closures; note that ``functions'' in CC
are implicit closures.
The typing rule for closures is:

\begin{mathpar}
  \inferrule*[right=\rulename{Clo}]
  {\ttyjudg{\tlenv}{\te}{\tcodety{\tn:\tApr,\tx:\tA}{\tB}} \\
  \ttyjudg{\tlenv}{\tepr}{\tApr}}
  {\ttyjudg{\tlenv}{\tcloe{\te}{\tepr}}{\tpity{\tx}{\subst{\tA}{\tepr}{\tn}}{\subst{\tB}{\tepr}{\tn}}}}
\end{mathpar}

\noindent We should think of a closure \im{\tcloe{\te}{\tepr}} not as a pair, but as a delayed partial
application of the code \im{\te} to its environment \im{\tepr}.
This intuition is formalized in the typing rule since the environment is substituted into the type,
just as in dependent-function application in CC.

To understand our translation, let us start with the translation of functions.

\begin{inlinedisplay}
  \begin{array}{@{\hspace{0pt}}r@{\hspace{1ex}}c@{\hspace{1ex}}l}
    (\sfune{\sx}{\sA}{\se})^+ &\defeq&
                                     \tcloe{(\tnfune{\begin{stackTL}(\tn:\tnsigmaty{(\txi:\sAi^+\dots)},\tx:\tlete{\tnpaire{\txi\dots}}{\tn}{\sA^+})}{
                                           \\
                                           \tlete{\tnpaire{\txi\dots}}{\tn}{\se^+}})}{\tnpaire{\txi\dots}}
                                     \end{stackTL} \\
                              && \text{ where \im{\sxi:\sAi\dots} are the free variables of \im{\se} and \im{\sA}}
  \end{array}
\end{inlinedisplay}

\noindent The translation of functions is simple to construct.
We know we want to produce a closure containing code and its environment.
We know the environment should be constructed from the free variables of the body of the function,
namely \im{\se}, and, due to dependent types, the type annotation \im{\sA}.

The question is: what should the type translation of \im{\sfontsym{\Pi}} types be?
Let's return to our polymorphic identity function (just the inner closure).
If we apply the above translation, we produce the following for the inner closure.
We know its type by following the typing rules \rulename{Clo} and \rulename{Code} above.

\begin{inlinedisplay}
  \begin{array}{l}
    \tcloe{\tnfune{(\tntwo:\tpairty{\tstarty}{\tunitty},\tx:\tfste{\tntwo})}{\tx}}
    {\tnpaire{\tA,\tnpaire{}}} :
    \\ \tpity{(\tx}{\subst{(\tfste{\tntwo})}{\tnpaire{\tA,\tunite}}{\tntwo})}{\subst{(\tfste{\tntwo})}{\tnpaire{\tA,\tunite}}{\tntwo}}
  \end{array}
\end{inlinedisplay}

We know that the code
\im{\tnfune{(\tntwo:\tpairty{\tstarty}{\tunitty},\tx:\tfste{\tntwo})}{\tx}} has type
\im{\tcodety{\tpairty{\tstarty}{\tunitty},\tx:\tfste{\tntwo}}{\tfste{\tntwo}}}.
Following \rulename{Clo}, we substitute the environment into this type, so we get:

\begin{inlinedisplay}
\tpity{(\tx}{\subst{(\tfste{\tntwo})}{\tnpaire{\tA,\tunite}}{\tntwo})}{\subst{(\tfste{\tntwo})}{\tnpaire{\tA,\tunite}}{\tntwo}}
\end{inlinedisplay}

\noindent So how do we translate the function type \im{\spity{\sx}{\sA}{\sA}} into the closure type
\im{\tpity{(\tx}{\subst{(\tfste{\tntwo})}{\tnpaire{\tA,\tunite}}{\tntwo})}{\subst{(\tfste{\tntwo})}{\tnpaire{\tA,\tunite}}{\tntwo}}}?
Note that this type reduces to \im{\tpity{\tx}{\tA}{\tA}}.
So by the rule \rulename{Conv}, we simply need to translate
\im{\spity{\sx}{\sA}{\sA}} to \im{\tpity{\tx}{\tA}{\tA}}!

The key translation rules are given below.

\begin{inlinedisplay}
  \begin{array}{@{\hspace{0pt}}r@{\hspace{1ex}}c@{\hspace{1ex}}l}
    (\spity{\sx}{\sA}{\sB})^+ &\defeq &\tpity{\tx}{\sA^+}{\sB^+}\\
    (\sfune{\sx}{\sA}{\se})^+ &\defeq&
                                     \tcloe{(\tnfune{\begin{stackTL}(\tn:\tnsigmaty{(\txi:\sAi^+\dots)},\tx:\tlete{\tnpaire{\txi\dots}}{\tn}{\sA^+})}{
                                           \\
                                           \tlete{\tnpaire{\txi\dots}}{\tn}{\se^+}})}{\tnpaire{\txi\dots}}
                                     \end{stackTL} \\
                              && \text{ where \im{\sxi:\sAi\dots} are the free variables of \im{\se} and \im{\sA}}
  \end{array}
\end{inlinedisplay}

A final challenge remains in the design of our target language: we need to know when two closures
are equivalent.
As we just saw, CC partially evaluates terms while type checking.
If two closures get evaluated while resolving type equivalence, we may inline a term into the
environment for one closure but not the other.
When this happens, two closures that were syntactically identical and thus equivalent become
inequivalent.
We discuss this problem in detail in \fullref[]{sec:cc}, but essentially we need to know when two
syntactically distinct closures are equivalent.
Our solution is simple: get rid of the closures and keep inlining things!
%Our solution, given below, is simple: closures be damned, keep inlining things!

\begin{mathpar}
  \inferrule
  {\tequivjudg{\tlenv,\tx:\tA}{\subst{\teone}{\teonepr}{\tn}}{\subst{\tetwo}{\tetwopr}{\tn}}}
  {\tequivjudg{\tlenv}{\tcloe{(\tnfune{(\tn:\tApr,\tx:\tA)}{\teone})}{\teonepr}}{\tcloe{(\tnfune{(\tn:\tApr,\tx:\tA)}{\tetwo})}{\tetwopr}}}
\end{mathpar}

Two closures are equivalent when we inline the environment, free variables or not, and run the body of
the code.
We leave the argument free, too.
We run the bodies of the code to normal forms, then compare the normal forms.
Recall that equivalence runs terms while \emph{type checking} and does not change the program, so the
free variables do no harm.

This equivalence essentially corresponds to an \(\eta\)-principle for closures.
From it, we can derive a normal form for closures \im{\tcloe{\te}{\tepr}} that says the
environment \im{\tepr} contains only free variables, \ie, \im{\tepr = \tnpaire{\txi\dots}}.

The above is an intuitive, declarative presentation, but is incomplete without additional rules.
We use an algorithmic presentation that is similar to the \(\eta\)-equivalence rules for functions
in CC, which we show in \fullref[]{sec:target}.
\newcommand{\FigCCTerm}[1][t]{
  \begin{figure*}[#1]
    \judgshape[~where~{\im{\styjudg{\slenv}{\se}{\st}}}]{\ccjudg{\slenv}{\se}{\st}{\te}}
    \begin{mathpar}
      \inferrule*[right=\rulename{CC-*}]
      {~}
      {\ccjudg{\slenv}{\sstarty}{\sboxty}{\tstarty}}

      \inferrule*[right=\rulename{CC-Var}]
      {~}
      {\ccjudg{\slenv}{\sx}{\sA}{\tx}}

      \inferrule*[right=\rulename{CC-Let}]
      {\ccjudg{\slenv}{\se}{\sA}{\te} \\
        \ccjudg{\slenv}{\sA}{\sU}{\tA} \\
        \ccjudg{\slenv,\sx:\sA}{\sepr}{\sB}{\tepr}}
      {\ccjudg{\slenv}{\salete{\sx}{\se}{\sA}{\sepr}}{\subst{\sB}{\se}{\sx}}{\talete{\tx}{\te}{\tA}{\tepr}}}

      \inferrule*[right=\rulename{CC-Prod-*}]
      {\ccjudg{\slenv}{\sA}{\sU}{\tA} \\
        \ccjudg{\slenv,\sx:\sA}{\sB}{\sstarty}{\tB}}
      {\ccjudg{\slenv}{\spity{\sx}{\sA}{\sB}}{\sstarty}{\tpity{\tx}{\tA}{\tB}}}

      \inferrule*[right=\rulename{CC-Prod-\(\square\)}]
      {\ccjudg{\slenv}{\sA}{\sU}{\tA} \\
        \ccjudg{\slenv,\sx:\sA}{\sB}{\sboxty}{\tB}}
      {\ccjudg{\slenv}{\spity{\sx}{\sA}{\sB}}{\sboxty}{\tpity{\tx}{\tA}{\tB}}}

      \inferrule*[right=\rulename{CC-Lam}]
      {\ccjudg{\slenv,\sx:\sA}{\se}{\sB}{\te} \\
        \ccjudg{\slenv}{\sA}{\sU}{\tA} \\
        \ccjudg{\slenv,\sx:\sA}{\sB}{\sU}{\tB} \\
        \sxi:\sAi\dots{}=\DFV{\slenv}{\sfune{\sx}{\sA}{\se},\spity{\sx}{\sA}{\sB}} \\
        \ccjudg{\slenv}{\sAi}{\sU}{\tAi} \dots}
      {\ccjudg{\slenv}{\sfune{\sx}{\sA}{\se}}{\spity{\sx}{\sA}{\sB}}{
          \tcloe{
            \begin{stackTL}
              (\tnfune{\begin{stackTL}(\tn:\tnsigmaty{(\txi:\tAi\dots)},\tx:\tlete{\tnpaire{\txi\dots}}{\tn}{\tA})}{
                  \\
                  \tlete{\tnpaire{\txi\dots}}{\tn}{\te}})}{
              \end{stackTL}
              \\\tdnpaire{\txi\dots}{\tnsigmaty{(\txi:\tAi\dots)}}}}
        \end{stackTL}}

      \inferrule*[right=\rulename{CC-App}]
      {\ccjudg{\slenv}{\seone}{\spity{\sx}{\sA}{\sB}}{\teone} \\
        \ccjudg{\slenv}{\setwo}{\sA}{\tetwo}}
      {\ccjudg{\slenv}{\sappe{\seone}{\setwo}}{\subst{\sB}{\setwo}{\sx}}{\tappe{\teone}{\tetwo}}}

      \inferrule*[right=\rulename{CC-Sig-*}]
      {\ccjudg{\slenv}{\sA}{\sstarty}{\tA} \\
        \ccjudg{\slenv,\sx:\sA}{\sB}{\sstarty}{\tB}}
      {\ccjudg{\slenv}{\ssigmaty{\sx}{\sA}{\sB}}{\sstarty}{\tsigmaty{\tx}{\tA}{\tB}}}

      \inferrule*[right=\rulename{CC-Sig-\(\square\)}]
      {\ccjudg{\slenv}{\sA}{\sboxty}{\tA} \\
        \ccjudg{\slenv,\sx:\sA}{\sB}{\sboxty}{\tB}}
      {\ccjudg{\slenv}{\ssigmaty{\sx}{\sA}{\sB}}{\sstarty}{\tsigmaty{\tx}{\tA}{\tB}}}

      \inferrule*[right=\rulename{CC-Fst}]
      {\ccjudg{\slenv}{\se}{\ssigmaty{\sx}{\sA}{\sB}}{\te}}
      {\ccjudg{\slenv}{\sfste{\se}}{\sA}{\tfste{\te}}}

      \inferrule*[right=\rulename{CC-Snd}]
      {\ccjudg{\slenv}{\se}{\ssigmaty{\sx}{\sA}{\sB}}{\te}}
      {\ccjudg{\slenv}{\ssnde{\se}}{\subst{\sB}{\sfste{\se}}{\sx}}{\tsnde{\te}}}

      \inferrule*[right=\rulename{CC-Conv}]
      {\ccjudg{\slenv}{\se}{\sA}{\te}}
      {\ccjudg{\slenv}{\se}{\sB}{\te}}
      \paperonly{\vspace{-1ex}}
    \end{mathpar}
    \judgshape[~where~{\im{\swf{\slenv}}}]{\ccenvjudg{\slenv}{\tlenv}}
    \begin{mathpar}
      \inferrule*[right=\rulename{W-Empty}]
      {~}
      {\ccenvjudg{\cdot}{\cdot}}

      \inferrule*[right=\rulename{W-Assum}]
      {\ccenvjudg{\slenv}{\tlenv} \\
       \ccjudg{\slenv}{\sA}{\_}{\tA}}
      {\ccenvjudg{\slenv,\sx:\sA}{\tlenv,\tx:\tA}}

      \inferrule*[right=\rulename{W-Def}]
      {\ccenvjudg{\slenv}{\tlenv} \\
        \ccjudg{\slenv}{\sA}{\_}{\tA} \\
        \ccjudg{\slenv}{\se}{\sA}{\te}}
      {\ccenvjudg{\slenv,\sx = \se:\sA}{\tlenv,\tx = \te:\tA}}
    \end{mathpar}
    \caption{Closure Conversion}
    \label{fig:cc:term}
  \end{figure*}
}

\newcommand{\FigCCCCSyntax}[1][t]{
  \begin{figure}[#1]
    \paperonly{\vspace{-1ex}}
    \begin{displaymath}
      \begin{array}{@{\hspace{0em}}l@{\hspace{1ex}}r@{\hspace{1ex}}c@{\hspace{1ex}}l}
        \techrptonly{
        \bnflabel{Universes} &
                               \tU & \bnfdef & \tstarty \bnfalt \tboxty
        \\[6pt]
        }
        \bnflabel{Expressions} & \te,\tA,\tB & \bnfdef &
                                                         \techrptonly{\tx
                                                         \bnfalt \tstarty
                                                         \bnfalt \talete{\tx}{\te}{\tA}{\te}
                                                         }\paperonly{\cdots}
                                                         \bnfalt \tunitty
                                                         \bnfalt \tunite
                                                         \bnfalt \tcodety{\txpr:\tApr,\tx:\tA}{\tB}
        \paperonly{\\ &&} \bnfalt \paperonly{&}
                                                         \tnfune{(\txpr:\tApr,\tx:\tA)}{\te}
        \techrptonly{\\[6pt] &&} \bnfalt \techrptonly{&}
                                          \tpity{\tx}{\tA}{\tB}
                                          \bnfalt \tcloe{\te}{\te}
                                          \techrptonly{
                                          \bnfalt \tappe{\te}{\tepr}
                                           \bnfalt \tsigmaty{\tx}{\tA}{\tB}
                                          \bnfalt \tfste{\te}
                                          \bnfalt \tsnde{\te}
                                          }
      \end{array}
    \end{displaymath}
    \paperonly{\vspace{1ex}}
    \caption{CC-CC Syntax\paperonly{ (excerpts)}}
    \label{fig:cc-cc:syntax}
  \end{figure}
}

\newcommand{\FigCCCCConv}[1][t]{
  \begin{figure}[#1]
    \judgshape{\tstepjudg[\step]{\tlenv}{\te}{\tepr}}
    \paperonly{\vspace{-3ex}}
    \begin{displaymath}
      \begin{array}{rcl@{\hspace{2ex}}l}
        \paperonly{&\!\!\!\!\vdots&\\}
        \techrptonly{
        \tx & \step_{\delta} & \te & \where{\tx = \te : \tA \in \tlenv}

        \\[4pt]
        \talete{\tx}{\te}{\tA}{\teone} & \step_{\zeta} & \subst{\teone}{\te}{\tx}
        \\[4pt]
        }
        \tappe{\tcloe{\tnfune{\txpr:\tApr,\tx:\tA}{\teone}}{\tepr}}{\te}
        & \step_{\beta} &
                          \subst{\subst{\teone}{\tepr}{\txpr}}{\te}{\tx}
        \techrptonly{
        \\[4pt]
        \tfste{\tpaire{\teone}{\tetwo}} & \step_{\pi_{1}} & \teone

        \\[4pt]
        \tsnde{\tpaire{\teone}{\tetwo}} & \step_{\pi_{2}} & \tetwo
        \\[4pt]
        }
      \end{array}
    \end{displaymath}
    \judgshape{\tequivjudg{\tlenv}{\te}{\tepr}}
    \begin{mathpar}
      \paperonly{\cdots}

      \techrptonly{
      \inferrule*[right=\rulename{\figm{\equiv}}]
        {\tstepjudg[\stepstar]{\tlenv}{\teone}{\te} \\
         \tstepjudg[\stepstar]{\tlenv}{\tetwo}{\te}}
        {\tequivjudg{\tlenv}{\teone}{\tetwo}}

      }
      \inferrule*[right=\rulename{\figm{\equiv}-Clo\figm{_1}}]
      {\tstepjudg[\stepstar]{\tlenv}{\teone}{\tcloe{\tnfune{(\txpr:\tApr,\tx:\tA)}{\teonepr}}{\tepr}}
        \\
        \tstepjudg[\stepstar]{\tlenv}{\tetwo}{\tetwopr} \\
        \tequivjudg{\tlenv,\tx:\tA}{\subst{\teone}{\tepr}{\txpr}}{\tappe{\tetwopr}{\tx}}}
        {\tequivjudg{\tlenv}{\teone}{\tetwo}}

        \inferrule*[right=\rulename{\figm{\equiv}-Clo\figm{_2}}]
        {\tstepjudg[\stepstar]{\tlenv}{\tetwo}{\tcloe{\tnfune{(\txpr:\tApr,\tx:\tA)}{\tetwopr}}{\tepr}}
          \\
          \tstepjudg[\stepstar]{\tlenv}{\teone}{\teonepr} \\
          \tequivjudg{\tlenv,\tx:\tA}{\tappe{\teonepr}{\tx}}{\subst{\tetwopr}{\tepr}{\txpr}}}
        {\tequivjudg{\tlenv}{\teone}{\tetwo}}
    \end{mathpar}
    \paperonly{\vspace{1ex}}
    \caption{CC-CC Conversion and Equivalence\paperonly{ (excerpts)}}
    \label{fig:cc-cc:conv}
    \paperonly{\vspace{1ex}}
  \end{figure}
}

\newcommand{\FigCCCCTyping}[1][t]{
  \begin{figure}[#1]
    \judgshape{\ttyjudg{\tlenv}{\te}{\tt}}
    \paperonly{\vspace{-1.5ex}}
    \begin{mathpar}
      \paperonly{\cdots}
      \techrptonly{
        \inferrule*[right=\rulename{Ax-*}]
        {\twf{\tlenv}}
        {\ttyjudg{\tlenv}{\tstarty}{\tboxty}}

        \inferrule*[right=\rulename{Var}]
        {\tx : \tA \in \tlenv \\
          \twf{\tlenv}}
        {\ttyjudg{\tlenv}{\tx}{\tA}}

        \inferrule*[right=\rulename{T-Unit}]
        {\twf{\tlenv}}
        {\ttyjudg{\tlenv}{\tunitty}{\tstarty}}

        \inferrule*[right=\rulename{Unit}]
        {\twf{\tlenv}}
        {\ttyjudg{\tlenv}{\tunite}{\tunitty}}

        \inferrule*[right=\rulename{Let}]
        {\ttyjudg{\tlenv}{\te}{\tA} \\
          \ttyjudg{\tlenv,\tx=\te:\tA}{\tepr}{\tB}}
        {\ttyjudg{\tlenv}{\talete{\tx}{\te}{\tA}{\tepr}}{\subst{\tB}{\te}{\tx}}}

        \inferrule*[right=\rulename{Prod-*}]
        {\ttyjudg{\tlenv}{\tA}{\tU} \\
          \ttyjudg{\tlenv,\tx:\tA}{\tB}{\tstarty}}
        {\ttyjudg{\tlenv}{\tpity{\tx}{\tA}{\tB}}{\tstarty}}

        \inferrule*[right=\rulename{Prod-\(\square\)}]
        {\ttyjudg{\tlenv}{\tA}{\tU} \\
          \ttyjudg{\tlenv,\tx:\tA}{\tB}{\tboxty}}
        {\ttyjudg{\tlenv}{\tpity{\tx}{\tA}{\tB}}{\tboxty}}
      }

      \inferrule*[right=\rulename{T-Code-*}]
      {\ttyjudg{\tlenv,\txpr:\tApr,\tx:\tA}{\tB}{\tstarty}}
      {\ttyjudg{\tlenv}{\tcodety{\txpr:\tApr,\tx:\tA}{\tB}}{\tstarty}}

      \inferrule*[right=\rulename{T-Code-\(\square\)}]
      {\ttyjudg{\tlenv,\txpr:\tApr,\tx:\tA}{\tB}{\tboxty}}
      {\ttyjudg{\tlenv}{\tcodety{\tx:\tA,\txpr:\tApr}{\tB}}{\tboxty}}

      \inferrule*[right=\rulename{Code}]
      {\ttyjudg{\cdot,\txpr:\tApr,\tx:\tA}{\te}{\tB}}
      {\ttyjudg{\tlenv}{\tnfune{(\txpr:\tApr,\tx:\tA)}{\te}}{\tcodety{\txpr:\tApr,\tx:\tA}{\tB}}}

      \inferrule*[right=\rulename{Clo}]
      {\ttyjudg{\tlenv}{\te}{\tcodety{\txpr:\tApr,\tx:\tA}{\tB}} \\
        \ttyjudg{\tlenv}{\tepr}{\tApr}}
      {\ttyjudg{\tlenv}{\tcloe{\te}{\tepr}}{\tpity{\tx}{\subst{\tA}{\tepr}{\txpr}}{\subst{\tB}{\tepr}{\txpr}}}}

      \techrptonly{
        \inferrule*[right=\rulename{App}]
        {\ttyjudg{\tlenv}{\te}{\tpity{\tx}{\tApr}{\tB}} \\
          \ttyjudg{\tlenv}{\tepr}{\tApr}}
        {\ttyjudg{\tlenv}{\tappe{\te}{\tepr}}{\subst{\tB}{\tepr}{\tx}}}

        \inferrule*[right=\rulename{Sig-*}]
        {\ttyjudg{\tlenv}{\tA}{\tstarty} \\
          \ttyjudg{\tlenv,\tx:\tA}{\tB}{\tstarty}}
        {\ttyjudg{\tlenv}{\tsigmaty{\tx}{\tA}{\tB}}{\tstarty}}

        \inferrule*[right=\rulename{Sig-\(\square\)}]
        {\ttyjudg{\tlenv}{\tA}{\tU} \\
          \ttyjudg{\tlenv,\tx:\tA}{\tB}{\tboxty}}
        {\ttyjudg{\tlenv}{\tsigmaty{\tx}{\tA}{\tB}}{\tboxty}}

        \inferrule*[right=\rulename{Fst}]
        {\ttyjudg{\tlenv}{\te}{\tsigmaty{\tx}{\tA}{\tB}}}
        {\ttyjudg{\tlenv}{\tfste{\te}}{\tA}}

        \inferrule*[right=\rulename{Snd}]
        {\ttyjudg{\tlenv}{\te}{\tsigmaty{\tx}{\tA}{\tB}}}
        {\ttyjudg{\tlenv}{\tsnde{\te}}{\subst{\tB}{\tfste{\te}}{\tx}}}

        \inferrule*[right=\rulename{Conv}]
        {\ttyjudg{\tlenv}{\te}{\tA} \\
          \ttyjudg{\tlenv}{\tB}{\tU} \\
          \tequivjudg{\tlenv}{\tA}{\tB}}
        {\ttyjudg{\tlenv}{\te}{\tB}}
      }

    \end{mathpar}
    \paperonly{\vspace{1ex}}
    \caption{CC-CC Typing\paperonly{ (excerpts)}}
    \label{fig:cc-cc:type}
  \end{figure}
}

\newcommand{\FigCCCCWF}[1][t]{
  \begin{figure}[#1]
    \judgshape{\twf{\tlenv}}
    \begin{mathpar}
        \inferrule*[right=\rulename{W-Empty}]
        {~}
        {\twf{\cdot}}

      \inferrule*[right=\rulename{W-Assum}]
        {\twf{\tlenv} \\
         \ttyjudg{\tlenv}{\tA}{\tU}}
        {\twf{\tlenv,\tx:\tA}}

      \inferrule*[right=\rulename{W-Def}]
        {\twf{\tlenv} \\
         \ttyjudg{\tlenv}{\te}{\tA} \\
         \ttyjudg{\tlenv}{\tB}{\tU}}
        {\twf{\tlenv,\tx = \te :\tA}}
    \end{mathpar}
    \caption{CC-CC well-formed environments}
    \label{fig:cc-cc:wf}
  \end{figure}
}

\newcommand{\FigModelCCCC}[1][t]{
  \begin{figure*}[#1]
    \judgshape{\mjudg{\tlenv}{\te}{\tA}{\se}}
    \paperonly{\vspace{-2ex}}
    \begin{mathpar}
      \paperonly{\cdots}
      \techrptonly{
        \inferrule*[right=\rulename{M-*}]
        {~}
        {\mjudg{\tlenv}{\tstarty}{\tboxty}{\sstarty}}

        \inferrule*[right=\rulename{M-Var}]
        {~}
        {\mjudg{\tlenv}{\tx}{\tA}{\sx}}

        \inferrule*[right=\rulename{M-T-Unit}]
        {~}
        {\mjudg{\tlenv}{\tunitty}{\tstarty}{\spity{\salpha}{\sstarty}{\spity{\sx}{\salpha}{\salpha}}}}

        \inferrule*[right=\rulename{M-Unit}]
        {~}
        {\mjudg{\tlenv}{\tunite}{\tunitty}{\sfune{\salpha}{\sstarty}{\sfune{\sx}{\salpha}{\sx}}}}

        \inferrule*[right=\rulename{M-Let}]
        {\mjudg{\tlenv}{\te}{\tA}{\se} \\
          \mjudg{\tlenv,\tx=\te:\tA}{\tepr}{\tB}}
        {\mjudg{\tlenv}{\talete{\tx}{\te}{\tA}{\tepr}}{\subst{\tB}{\te}{\tx}}{\salete{\sx}{\se}{\sA}{\sepr}}}
      }

      \inferrule*[right=\rulename{M-Prod-*}]
      {\mjudg{\tlenv}{\tA}{\tU}{\sA} \\
        \mjudg{\tlenv,\tx:\tA}{\tB}{\tstarty}{\sB}}
      {\mjudg{\tlenv}{\tpity{\tx}{\tA}{\tB}}{\tstarty}{\spity{\sx}{\sA}{\sB}}}

      \techrptonly{
        \inferrule*[right=\rulename{M-Prod-\(\square\)}]
        {\mjudg{\tlenv}{\tA}{\tU}{\sA} \\
          \mjudg{\tlenv,\tx:\tA}{\tB}{\tboxty}{\sB}}
        {\mjudg{\tlenv}{\tpity{\tx}{\tA}{\tB}}{\tboxty}{\spity{\sx}{\sA}{\sB}}}

      }
      \inferrule*[right=\rulename{M-T-Code-*}]
      {\mjudg{\tlenv}{\tApr}{\tUpr}{\sApr} \\
        \mjudg{\tlenv,\txpr:\tApr}{\tA}{\tU}{\sA} \\
      \mjudg{\tlenv,\txpr:\tApr,\tx:\tA}{\tB}{\tstarty}{\sB}}
      {\mjudg{\tlenv}{\tcodety{\txpr:\tApr,\tx:\tA}{\tB}}{\tstarty}{\spity{\sxpr}{\sApr}{\spity{\sx}{\sA}{\sB}}}}

      \inferrule*[right=\rulename{M-T-Code-\(\square\)}]
      {\mjudg{\tlenv}{\tApr}{\tUpr}{\sApr} \\
        \mjudg{\tlenv,\txpr:\tApr}{\tA}{\tU}{\sA} \\
      \mjudg{\tlenv,\txpr:\tApr,\tx:\tA}{\tB}{\tboxty}{\sB}}
      {\mjudg{\tlenv}{\tcodety{\txpr:\tApr,\tx:\tA}{\tB}}{\tboxty}{\spity{\sxpr}{\sApr}{\spity{\sx}{\sA}{\sB}}}}

      \inferrule*[right=\rulename{M-Code}]
      {\mjudg{\tlenv}{\tApr}{\tUpr}{\sApr} \\
        \mjudg{\tlenv,\txpr:\tApr}{\tA}{\tU}{\sA} \\
        \mjudg{\tlenv,\txpr:\tApr,\tx:\tA}{\tB}{\tU}{\sB} \\
        \mjudg{\tlenv,\txpr:\tApr,\tx:\tA}{\te}{\tB}{\se}}
      {\mjudg{\tlenv}{\tnfune{(\txpr:\tApr,\tx:\tA)}{\te}}{\tcodety{\txpr:\tApr,\tx:\tA}{\tB}}
        {\sfune{\sxpr}{\sApr}{\sfune{\sx}{\sA}{\se}}}}

      \inferrule*[right=\rulename{M-Clo}]
      {\mjudg{\tlenv}{\te}{\tcodety{\txpr:\tApr,\tx:\tA}{\tB}}{\se} \\
      \mjudg{\tlenv}{\tepr}{\tApr}{\sepr}}
      {\mjudg{\tlenv}{\tcloe{\te}{\tepr}}{\tpity{\tx}{\subst{\tA}{\tepr}{\tx}}{\subst{\tB}{\tepr}{\tx}}}{\sappe{\se}{\sepr}}}

      \inferrule*[right=\rulename{M-App}]
      {\mjudg{\tlenv}{\te}{\tpity{\tx}{\tA}{\tB}}{\se} \\
        \mjudg{\tlenv}{\tepr}{\tA}{\sepr}}
      {\mjudg{\tlenv}{\tappe{\te}{\tepr}}{\subst{\tB}{\tepr}{\tx}}{\sappe{\se}{\sepr}}}

      \techrptonly{
        \inferrule*[right=\rulename{M-Sig-*}]
        {\mjudg{\tlenv}{\tA}{\tstarty}{\sA} \\
          \mjudg{\tlenv,\tx:\tA}{\tB}{\tstarty}{\sB}}
        {\mjudg{\tlenv}{\tsigmaty{\tx}{\tA}{\tB}}{\tstarty}{\ssigmaty{\sx}{\sA}{\sB}}}

        \inferrule*[right=\rulename{M-Sig-\(\square\)}]
        {\mjudg{\tlenv}{\tA}{\tU}{\sA} \\
          \mjudg{\tlenv,\tx:\tA}{\tB}{\tboxty}{\sB}}
        {\mjudg{\tlenv}{\tsigmaty{\tx}{\tA}{\tB}}{\tboxty}{\ssigmaty{\sx}{\sA}{\sB}}}

        \inferrule*[right=\rulename{M-Fst}]
        {\mjudg{\tlenv}{\te}{\tsigmaty{\tx}{\tA}{\tB}}{\se}}
        {\mjudg{\tlenv}{\tfste{\te}}{\tA}{\sfste{\se}}}

        \inferrule*[right=\rulename{M-Snd}]
        {\mjudg{\tlenv}{\te}{\tsigmaty{\tx}{\tA}{\tB}}{\se}}
        {\mjudg{\tlenv}{\tsnde{\te}}{\tA}{\ssnde{\se}}}

        \inferrule*[right=\rulename{M-Conv}]
        {\mjudg{\tlenv}{\te}{\tB}{\se}}
        {\mjudg{\tlenv}{\te}{\tA}{\se}}
      }
    \end{mathpar}
    \paperonly{\vspace{1ex}}
    \caption{Translation from CC-CC to CC\paperonly{ (excerpts)}}
    \label{fig:model}
    \paperonly{\vspace{-1ex}}
  \end{figure*}
}

\section{Target: CC, Closure-Converted (CC-CC)}
\label{sec:target}

\FigCCCCSyntax
The target language CC-CC is based on CC, but first-class functions are replaced by closed
code and closures.
We add a primitive unit type \im{\tunitty} to support encoding environments.
\techrptonly{This language is typeset in a \tfonttext{bold, red, serif font}.}%
We extend the syntax of expressions, \fullref[]{fig:cc-cc:syntax}, with a unit value \im{\tnpaire{}}
and its type \im{\tunitty}, closed code \im{\tnfune{\tn:\tApr,\tx:\tA}{\te}} and dependent code types
\im{\tcodety{\tn:\tApr,\tx:\tA}{\tB}}, and closure values \im{\tcloe{\te}{\tepr}} and dependent
closure types \im{\tpity{\tx}{\tA}{\tB}}.
The syntax of application \im{\tappe{\te}{\tepr}} is unchanged, but it now applies closures
instead of functions.

We define additional syntactic sugar for sequences of terms, to support writing environments whose
length is arbitrary.
We write a sequence of terms \im{\tei\dots} to mean a sequence of length \im{\len{i}} of expressions
\im{\tein{i_0},\dots,\tein{i_n}}.
We extend the notation to patterns such as \im{\txi:\tAi\dots}, which implies two sequences
\im{\txin{i_0},\dots,\txin{i_n}} and \im{\tAin{0},\dots,\tAin{i_n}} each of length \im{\len{i}}.
We define environments as dependent n-tuples, written
\im{\tdnpaire{\tei\dots}{\tnsigmaty{(\txi:\tAi\dots)}}}.
We encode dependent n-tuples as nested dependent pairs followed by a unit value, \ie,
\im{\tpaire{\tein{0}}{\tpaire{\dots}{\tpaire{\tein{i}}{\tunite}}}}.
We omit the annotation on n-tuples \im{\tnpaire{\tei\dots}} when it is obvious from context.
We also define pattern matching on n-tuples, written \im{\tlete{\tnpaire{\txi\dots}}{\tepr}{\te}}, to
perform the necessary nested projections, \ie,
\im{\tlete{\txin{0}}{\tfste{\tepr}}{\dots\tlete{\txi}{\tfste{\tsnde{\dots \tsnde{\tepr}}}}{\te}}}.

\FigCCCCConv
In \fullref[]{fig:cc-cc:conv} we present the additional conversion and equivalence rules for CC-CC.
Code cannot be applied directly, but must be part of a closure.
Closures applied to an argument \(\beta\)-reduce, applying the underlying code to the environment
and the argument.
All the other conversion rules remain unchanged.
For equivalence, we no longer have the usual \(\eta\) rules, since functions have been turned into
closures.
Instead, we need \(\eta\) rules for closures.

\FigCCCCTyping
We give the typing rules in \fullref[]{fig:cc-cc:type}.
\paperonly{All unspecified rules are unchanged from the source language. }%
\techrptonly{Most rules are unchanged from the source language. }%
The most interesting rule is \rulename{Code}, which
that code only type checks when it is closed.
This rule captures the entire point of typed closure conversion and gives us static machine-checked
guarantees that our translation produces closed code.
The typing rule \rulename{Clo} for closures \im{\tcloe{\te}{\tepr}} substitutes the environment
\im{\tepr} into the type of the closure\paperonly{, as discussed in \fullref[]{sec:idea}}.
This is similar to the CC rule \rulename{App} that substitutes a function argument into the result
type of a function.
\paperonly{As we discussed in \fullref[]{sec:idea}, t}\techrptonly{T}his is also critical to type
preservation, since our translation must generate closure types with free variables and then
synchronize the closure type containing free variables with a closed code type.
As with \im{\sfontsym{\Pi}} types in CC, we have two rules for well typed \im{\tfont{Code}} types.
The rule \rulename{T-Code-*} allows impredicativity in \im{\tstarty}, while \rulename{T-Code-\im{\square}}
is predicative.

\subsection{Type Safety and Consistency}
\label{sec:consistency}
\FigModelCCCC
We prove that CC-CC is type safe when interpreted as a programming language and consistent when
interpreted as a logic.
Type safety guarantees that all programs in CC-CC have well-defined behavior, and consistency ensures
that when interpreting types as propositions and programs as proofs, we cannot prove
\im{\tfont{False}} in CC-CC.
We prove both theorems by giving a model of CC-CC in CC, \ie, by encoding the target language in the
source language.
The model reduces type safety and consistency of CC-CC to that of CC, which is known to be type safe
and consistent.
This standard technique is well explained by \citet{boulier2017}.

We construct a model essentially by ``decompiling'' closures, translating code to functions and
closures to partial application.
To show this translation is a model, we need to show that it preserves falseness---\ie, that we
translate \im{\tfont{False}} to \im{\sfont{False}}---and show that the translation is
type-preserving---\ie, we translate any well-typed CC-CC program (valid proof) into a well-typed
program in CC.
To extend the model to type safety, we must also show that the translation preserves reduction
semantics---\ie, that reducing an expression in CC-CC is essentially equivalent to reducing the
translated term in CC.
Since our type system includes reduction, we already prove this to show type preservation.

We then prove consistency and type safety of CC-CC by contradiction.
If CC-CC were inconsistent, then we could prove the proposition \im{\tfont{False}} in CC-CC, and
translate that proof into a valid proof of \im{\sfont{False}} in CC.
But since CC is consistent, we can never produce a proof of \im{\sfont{False}} in CC, therefore we
could not have constructed one in CC-CC.
A similar argument applies for type safety.
Since we preserve reduction semantics in CC-CC, if a term had undefined behavior, we could translate
the term into a CC term with undefined behavior.
However, CC has no terms with undefined behavior, hence neither does CC-CC.

The translation from CC-CC to CC, \fullref[]{fig:model}, is defined on typing derivations.
We use the following notation.

\begin{inlinedisplay}
  \te^\circ \defeq \se~\where{\mjudg{\tlenv}{\te}{\tA}{\se}}
\end{inlinedisplay}

\noindent The CC expression \im{\te^\circ} refers to the expression produced by translating the CC-CC expression
\im{\te}, with the typing derivation for \im{\te} as an implicit argument.

The rule \rulename{M-Code} translates a code type \im{\tcodety{\tn:\tApr,\tx:\tA}{\tB}} to
the curried function type \im{\spity{\sn}{{\tApr}^\circ}{\spity{\tx}{\tA^\circ}{\tB^\circ}}}.
The rule \rulename{M-Code} models code \im{\tnfune{\tn:\tApr,\tx:\tA}{\te}} as a curried function
\im{\sfune{\sn}{{\tApr}^\circ}{\sfune{\sx}{\tA^\circ}{\te^\circ}}}.
Observe that the inner function produced in CC is not closed, but that is not a problem since the
model only exists to prove type safety and consistency.
It is only in CC-CC programs that code must be closed.
The rule \rulename{M-Clo} models a closure \im{\tcloe{\te}{\tepr}} as the application
\im{\sappe{\te^\circ}{{\tepr}^\circ}}---\ie, the application of the function \im{\te^\circ} to
its environment \im{{\tepr}^\circ}.
We model \im{\tfont{Unit}}\paperonly{, omitted for brevity,} with the standard Church encoding as the
polymorphic identity function.
All other rules simply recursively translate subterms.

We first prove that this translation preserves falseness.
We encode \im{\tfont{False}} in CC-CC as \im{\tpity{\tA}{\tstarty}{\tA}}.
This encoding represents a function that takes any arbitrary proposition \im{\tA} and
returns a proof of \im{\tA}.
Similar, in CC \im{\sfont{False}} as \im{\spity{\sA}{\sstarty}{\sA}}.
It is clear from \rulename{M-Prod-*} that the translation preserves falseness.
We use \im{=} as the terms are not just definitionally equivalent, but syntactically identical.

\begin{lemma}[False Preservation]
  \label{sec:m:false-pres}
  \im{\tfont{False}^\circ = \sfont{False}}
\end{lemma}

To prove type preservation, we split the proof into three key lemmas.
First, we show \emph{compositionality}, \ie, that the translation from CC-CC to CC commutes with
substitution.
Then we prove preservation of reduction semantics and equivalence, which essentially follows from
compositionality.
Finally, we prove type preservation, which relies on preservation of equivalence and on
compositionality.
\paperonly{The proofs are straightforward, since the typing rules in CC-CC essentially
  correspond to partial application already, so we elide them here.
  They follow the same structure as our type preservation proof for closure conversion, which we
  present in \fullref[]{sec:cc}.
  For complete details, see our \material.}

Compositionality is an important lemma since the type system and conversion relations are defined by
substitution.
\begin{lemma}[Compositionality]
  \label{lem:m:subst}
  \im{(\subst{\te}{\tepr}{\tx})^\circ = \subst{\te^{\circ}}{\te^{\tprime\circ}}{\sx}}
\end{lemma}
\techrptonly{
\begin{proof}
  The proof is by induction on the typing derivation \nonbreaking{\im{\ttyjudg{\tlenv}{\te}{\tA}}}.
  Recall that by convention this derivation is an implicit argument to the lemma.
  \paperonly{The proof is straightforward.\qedhere}
    \begin{itemize}
      \case \rulename{Ax-*}
      Trivial, since \im{\te = \tstarty} cannot have free variables.

      \case \rulename{Var}
      Hence \im{\te = \txpr}.
      There are two subcases:
      \scase \im{\txpr = \tx}
      Then the proof follows since \im{(\subst{\tx}{\tepr}{\tx})^\circ = \te^{\tprime\circ} = \subst{\sx}{\te^{\tprime\circ}}{\sx}}

      \scase \im{\txpr \neq \tx}
      Then the proof follows since \im{(\subst{\txpr}{\tepr}{\tx})^\circ = \sxpr = \subst{\sxpr}{\te^{\tprime\circ}}{\sx}}

      \case \rulename{Let}
      Follows easily by the inductive hypotheses, since both the translation of \im{\tfont{let}} and
      the definition of substitution are structural, except for the capture avoidance reasoning.

      \case \rulename{Prod-*}
      Follows easily by the inductive hypotheses, since both the translation of \im{\tfontsym{\Pi}} and
      the definition of substitution are structural, except for the capture avoidance reasoning.

      \item[\vdots]

      \case \rulename{Conv}
      Recall that the translation is defined by induction on typing derivations, and therefore we have
      the conversion typing rule:

      \begin{mathpar}
        \inferrule*[right=\rulename{Conv}]
        {\ttyjudg{\tlenv}{\te}{\tA} \\
          \ttyjudg{\tlenv}{\tB}{\tU} \\
        \tequivjudg{\tlenv}{\tA}{\tB}}
        {\ttyjudg{\tlenv}{\te}{\tB}}
      \end{mathpar}

      We must show that \im{(\subst{\te}{\tepr}{\tx})^\circ \equiv
        \subst{\te^\circ}{{\tepr}^\circ}{\sx}} at, loosely speaking, the type \im{\tB^\circ}.
      (Loosely, since we haven't show type preservation yet.)

      By the induction hypothesis applied to \im{\ttyjudg{\tlenv}{\te}{\tA}}, we have that
      \im{(\subst{\te}{\tepr}{\tx})^\circ \equiv
        \subst{\te^\circ}{{\tepr}^\circ}{\sx}} at the type \im{\tA^\circ}.

      The astute type theorist may be concerned that we first need to show that these terms are
      well-typed---\ie, that we need to show type preservation---and that \im{\tA^\circ \equiv
        \tB^\circ}, \ie, \emph{coherence}.
      However, our definition of equivalence in CC and CC-CC is based on the CIC \emph{untyped}
      equivalence~\cite[Chapter 4]{coq2017}, so the proof is already done.
      We can think of this equivalence as justifying semantic equivalences that are statically ruled
      out by a conservative syntactic type system.
      The advantage of this equivalence is that it allows us to stage the proof as we have.
    \end{itemize}
\end{proof}
  }

Next we show that the translation preserves reduction, or that our model in CC weakly simulates
reduction in CC-CC.
This is used both to show that equivalence is preserved, since equivalence is defined by reduction,
and to show type safety.
\begin{lemma}[Pres. of Reduction]
  \label{lem:m:red}
  \nonbreaking{If \im{\te \step \tepr} then \im{\te^\circ \stepstar \te^{\tprime\circ}}}
\end{lemma}
\techrptonly{
\begin{proof}
  By cases on \im{\te \step \tepr}.
  The only interesting case is for the reduction of closures.
  \begin{itemize}
    \case \im{\tappe{\tcloe{(\tnfune{\txpr:\tApr,\tx:\tA}{\tein{b}})}{\tepr}}{\te} \step_\beta \subst{\subst{\tein{b}}{\tepr}{\txpr}}{\te}{\tx}}

    \noindent We must show that

    \im{(\tappe{\tcloe{(\tnfune{\txpr:\tApr,\tx:\tA}{\tein{b}})}{\tepr}}{\te})^\circ \stepstar
      (\subst{\subst{\tein{b}}{\tepr}{\txpr}}{\te}{\tx})^\circ}
    \begin{align}
      & (\tappe{\tcloe{(\tnfune{\txpr:\tApr,\tx:\tA}{\tein{b}})}{\tepr}}{\te})^\circ \\
      &~= \sappe{(\sappe{(\sfune{\sxpr}{\tA^{\tprime\circ}}{\sfune{\sx}{\tA^\circ}}{\tein{b}^\circ})}{\te^{\tprime\circ}})}{\te^\circ} & \text{by definition} \\
      &~\step_\beta^2 \subst{\subst{\tein{b}^\circ}{\te^{\tprime\circ}}{\sxpr}}{\te^{\circ}}{\sx} \\
      &~= (\subst{\subst{\tein{b}}{\tepr}{\txpr}}{\te}{\tx})^\circ & \text{by \fullref[]{lem:m:subst}}
    \end{align}
  \end{itemize}
\end{proof}
}

Now we show that reduction \emph{sequences} are preserved.
This essentially follows from preservation of single-step reduction, \fullref[]{lem:m:red}.
\begin{lemma}[Preservation of Reduction Sequences]
  \label{lem:m:red*}
  If \im{\te \stepstar \tepr} then \im{\te^\circ \stepstar \te^{\tprime\circ}}
\end{lemma}
\techrptonly{
\begin{proof}
  The proof is by induction on the length of the reduction sequence \im{\te \stepstar \tepr}.  The
  base case is trivial, and the inductive case follows by \fullref{lem:m:red} and the inductive
  hypothesis.
\end{proof}
}

Next, we show \emph{coherence}, \ie, that the translation preserves equivalence.
The proof essentially follows from \fullref[]{lem:m:red*}, but we must show that our \(\eta\)
rule for closures is preserved.
\begin{lemma}[Coherence]
  \label{lem:m:coherence}
  If \im{\teone \equiv \tetwo} then \im{\teone^\circ \equiv \tetwo^{\circ}}
\end{lemma}
\techrptonly{
\begin{proof}
  The proof is by induction on the derivation \im{\te \equiv \tepr}.
  The only interesting case is for \(\eta\) equivalence of closures.
  \begin{itemize}
    \techrptonly{
      \case \rulename{\im{\equiv}}
      Follows by \fullref{lem:m:red*}.
    }
    \case \rulename{\im{\equiv}-Clo\im{_1}}

    \noindent By assumption, we have the following.
    \begin{enumerate}
    \item \im{\teone \stepstar {\tcloe{\tnfune{(\txpr:\tApr,\tx:\tA)}{\teonepr}}{\tepr}}}
    \item \im{\tetwo \stepstar \tetwopr}
    \item \im{\subst{\teone}{\tepr}{\txpr} \equiv \tappe{\tetwopr}{\tx}}
    \end{enumerate}

    \noindent We must show that \im{\teone^\circ \equiv \tetwo^\circ}.
    By \rulename{\im{\equiv}-\im{\eta_1}}, it suffices to show:
    \begin{enumerate}
        \item \im{\teone^\circ \stepstar
            \sfune{\sx}{\subst{\tA^\circ}{\te^{\tprime\circ}}{\sxpr}}{\subst{\teone^{\tprime\circ}}{\te^{\tprime\circ}}{\sxpr}}},
          which follows since:
          \begin{align}
            \qquad\teone^\circ &~\stepstar (\tcloe{\tnfune{(\txpr:\tApr,\tx:\tA)}{\teonepr}}{\tepr})^\circ & \text{by \fullref[]{lem:m:red*}}\\
            &~= \sappe{(\sfune{\sxpr}{\tA^{\tprime\circ}}{\sfune{\sx}{\tA^\circ}}{\teone^{\tprime\circ}})}{\te^{\tprime\circ}} \\
            &~\step {\sfune{\sx}{\subst{\tA^{\tprime\circ}}{\te^{\tprime\circ}}{\sxpr}}
              {\subst{\teone^{\tprime\circ}}{\te^{\tprime\circ}}{\sxpr}}}
            \end{align}
        \item \im{\tetwo^\circ \stepstar \tetwo^{\tprime\circ}} which follows by \fullref[]{lem:m:red*}.
        \item \im{{\subst{\teone^{\tprime\circ}}{\te^{\tprime\circ}}{\sxpr}} \equiv
            \sappe{{\tetwopr}^\circ}{\sx}}, which follows by the inductive hypothesis
          applied to \im{\subst{\teone}{\tepr}{\txpr} \equiv \tappe{\tetwopr}{\tx}} and \fullref[]{lem:m:subst}.\paperonly{\qedhere}
    \end{enumerate}
    \techrptonly{\case \rulename{\im{\equiv}-Clo\im{_2}} is symmetric.}
  \end{itemize}
\end{proof}
}

We can now show our final lemma: type preservation.
\begin{lemma}[Type Preservation]
  ~
  \label{lem:m:type-pres}
  \begin{enumerate}
    \item If \im{\twf{\tlenv}} then \im{\swf{\tlenv^{\circ}}}
    \item If \im{\ttyjudg{\tlenv}{\te}{\tA}} then \im{\styjudg{\tlenv^\circ}{\te^\circ}{\tA^\circ}}
  \end{enumerate}
\end{lemma}
\techrptonly{
\begin{proof}
  We prove parts 1 and 2 simultaneously by induction on the mutually defined judgments
  \im{\twf{\tlenv}} and \im{\ttyjudg{\tlenv}{\te}{\tA}}.
  Most cases follow easily by the induction hypothesis.
  \begin{itemize}
    \techrptonly{
      \case \rulename{W-Empty} Trivial.

      \case \rulename{W-Def}
      We must show that \im{\swf{(\tlenv,\tx = \te : \tA)}^\circ}.
      By \rulename{W-Def} in CC and part 1 of the inductive hypothesis, it suffices to show that
      \im{\styjudg{\tlenv^\circ}{\te^\circ}{\tA^\circ}}, which follows by part 2 of the inductive
      hypothesis applied to \im{\ttyjudg{\tlenv}{\te}{\tA}}.

      \case \rulename{W-Assum}
      We must show that \im{\swf{(\tlenv,\tx : \tA)}^\circ}.
      By \rulename{W-Assum} in CC and part 1 of the inductive hypothesis, it suffices to show that
      \im{\styjudg{\tlenv^\circ}{\tA^\circ}{\tU^\circ}}, which follows by part 2 of the inductive
      hypothesis applied to \im{\ttyjudg{\tlenv}{\tA}{\tU}}.

      \case \rulename{Ax-*}
      It suffices to show that \im{\swf{\tlenv^\circ}}, since \im{\tstarty^\circ = \sstarty},
      which follows by part 1 of the inductive hypothesis.

      \item[\vdots]
    }
    \case \rulename{T-Code-*}

    \noindent We have that
    \begin{inlinedisplay}
      \inferrule
      {\ttyjudg{\tlenv}{\tApr}{\tUpr} \\
        \ttyjudg{\tlenv,\txpr:\tApr}{\tA}{\tU} \\
        \ttyjudg{\tlenv,\txpr:\tApr,\tx:\tA}{\tB}{\tstarty}}
      {\ttyjudg{\tlenv}{\tcodety{\txpr:\tApr,\tx:\tA}{\tB}}{\tstarty}}
    \end{inlinedisplay}

    \noindent We must show that
    \im{\styjudg{\tlenv^\circ}{\spity{\sxpr}{\tA^{\tprime\circ}}{\spity{\sx}{\tA^\circ}{\tB^\circ}}}{\sstarty}}

    \noindent By two applications of \rulename{Prod-*}, it suffices to show
    \begin{itemize}
      \item \im{\styjudg{\tlenv^\circ}{\tA^{\tprime\circ}}{\tU^{\tprime\circ}}}, which follows by part
        2 of the inductive hypothesis.
      \item \im{\styjudg{\tlenv^\circ,\sxpr:\tA^{\tprime\circ}}{\tA^\circ}{\tU^\circ}},
        which follows by part 2 of the inductive hypothesis.
      \item \im{\styjudg{\tlenv^\circ,\sxpr:\tA^{\tprime\circ},\sx:\tA^\circ}{\tB^\circ}{\sstarty}},
        which follows by part 2 of the inductive hypothesis and by definition that \im{\tstarty^\circ
          = \sstarty}
    \end{itemize}

    \case \rulename{Code}

    \noindent We have that
    \begin{inlinedisplay}
      \inferrule
      {\ttyjudg{\tlenv,\txpr:\tApr,\tx:\tA}{\te}{\tB}}
      {\ttyjudg{\tlenv}{\tnfune{\txpr:\tApr,\tx:\tA}{\te}}{\tcodety{\txpr:\tApr,\tx:\tA}{\tB}}}
    \end{inlinedisplay}

    \noindent By definition of the translation, we must show
    \im{\styjudg{\tlenv^\circ}{\sfune{\sxpr}{\tA^{\tprime\circ}}{\sfune{\sx}{\tA^\circ}{\te^\circ}}}
      {\spity{\sxpr}{\tA^{\tprime\circ}}{\spity{\sx}{\tA^\circ}{\tB^\circ}}}}, which follows by
    two uses of \rulename{Lam} in CC and part 2 of the inductive hypothesis.

    \case \rulename{Clo}

    \noindent We have that
    \begin{inlinedisplay}
      \inferrule
      {\ttyjudg{\tlenv}{\te}{\tcodety{\txpr:\tApr,\tx:\tA}{\tB}} \\
        \ttyjudg{\tlenv}{\tepr}{\tApr}}
      {\ttyjudg{\tlenv}{\tcloe{\te}{\tepr}}{\tpity{\tx}{\subst{\tA}{\tepr}{\txpr}}{\subst{\tB}{\tepr}{\txpr}}}}
    \end{inlinedisplay}

    \noindent By definition of the translation, we must show that
    \im{\styjudg{\tlenv^\circ}{\sappe{\te^\circ}{\te^{\tprime\circ}}}{({\tpity{\tx}{\subst{\tA}{\tepr}{\txpr}}{\subst{\tB}{\tepr}{\txpr}}})^\circ}}.

    \noindent By \fullref{lem:m:subst}, it suffices to show that
    \im{\styjudg{\tlenv^\circ}{\sappe{\te^\circ}{\te^{\tprime\circ}}}{\spity{\sx}{\subst{\tA^\circ}{\te^{\tprime\circ}}{\sxpr}}{\subst{\tB^\circ}{\te^{\tprime\circ}}{\sxpr}}}}.

    \noindent By \rulename{App} in CC, it suffices to show that
    \begin{itemize}
    \item
      \im{\styjudg{\tlenv^\circ}{\te^\circ}{\spity{\sxpr}{\sApr}{\spity{\sx}{\tA^\circ}{\tB^\circ}}}},
      which follows with \im{\sApr = \tA^{\tprime\circ}} by part 2 of the inductive hypothesis.
    \item \im{\styjudg{\tlenv^\circ}{\te^{\tprime\circ}}{\sApr}}, which follows by part 2 of the
      inductive hypothesis.\paperonly{\qedhere}
    \end{itemize}
    \techrptonly{
      \case \rulename{App}
      Similar to the case for \rulename{Clo}.

      \case \rulename{Conv}
      Follows by part 2 of the inductive hypothesis and \fullref{lem:m:coherence}.
    }
  \end{itemize}
\end{proof}
}

\FigCCTerm[!th]

Finally, we can prove the desired consistency and type safety theorems.
\begin{theorem}[Consistency of CC-CC]
  \label{thm:m:sound}
  There does not exist a closed expression \im{\te} such that \im{\ttyjudg{\cdot}{\te}{\tfont{False}}}.
\end{theorem}

Type safety tells us that there is no undefined behavior that causes a program to get stuck before it
produces a value, and all programs terminate.
\begin{theorem}[Type Safety of CC-CC]
  \label{thm:m:safe}
  If \im{\ttyjudg{\cdot}{\te}{\tA}}, then \im{\te \stepstar \tv} and \im{\tv \not\step \tvpr}.
\end{theorem}

\newcommand{\FigDFVs}[1][t]{
  \begin{figure}[#1]
    \paperonly{\vspace{-3ex}}
    \begin{displaymath}
      \begin{array}{rcl}
        \DFV{\slenv}{\se,\sB} &\defeq
        &\slenvin{0},\dots,\slenvin{n},\sxin{0}:\sAin{0},\dots,\sxin{n}:\sAin{n} \\
        &where & \sxin{0},\dots,{}\sxin{n} = \FV{(\se,\sB)} \\
                          && \styjudg{\slenv}{\sxin{0}}{\sAin{0}},\dots,\styjudg{\slenv}{\sxin{n}}{\sAin{n}}\\
                          && \slenvin{0} = \DFV{\slenv}{\sAin{0},\_} \\
                          && \qquad\!\!\!\! \vdots \\
                          && \slenvin{n} = \DFV{\slenv}{\sAin{n},\_}
      \end{array}
    \end{displaymath}
    \paperonly{\vspace{1ex}}
    \caption{CC Dependent Free Variable Sequences}
    \label{fig:cc:dfvs}
    \paperonly{\vspace{-1.5ex}}
  \end{figure}
}

\section{Closure Conversion}
{
    \allowdisplaybreaks
\label{sec:cc}
We present the closure conversion translation in \fullref[]{fig:cc:term}.
We define the following notation for this translation.

\begin{inlinedisplay}
  \se^+ \defeq \te~\where{\ccjudg{\slenv}{\se}{\sA}{\te}}
\end{inlinedisplay}

\noindent The CC-CC expression \im{\se^+} refers to the translation of the well-typed CC term \im{\se},
with typing derivation for \im{\se} as an implicit parameter.

Every case of the translation except for functions is trivial,
including application \rulename{CC-App}, since application is still the
elimination form for closures after closure conversion.
In the nontrivial case \rulename{CC-Lam}, we translate CC functions to CC-CC closures\paperonly{, as
described in \fullref[]{sec:idea}}.
The translation of a function \im{\sfune{\sx}{\sA}{\se}} produces a closure
\im{\tcloe{\teone}{\tetwo}}.
We compute the free variables (and their type annotations) of the function \im{\sfune{\sx}{\sA}{\se}},
\im{\sxi:\sAi\dots}, using the metafunction
\im{\DFV{\slenv}{\sfune{\sx}{\sA}{\se},\spity{\sx}{\sA}{\sB}}} defined shortly.
The first component \im{\teone} is closed code.
Ignoring the type annotation for a moment, the code
\im{\tnfune{(\tn,\tx)}{\tlete{\tnpaire{\txi\dots}}{\tn}{\se^+}}} projects each of the \im{\len{i}}
free variables \im{\txi\dots} from the environment \im{\tn} and binds them in the scope of the body \im{\se^+}.
But CC-CC is dependently typed, so we also bind the free variables from the environment in the
type annotation for the argument \im{\tx}, \ie, producing the annotation
\im{\tx:\tlete{\tnpaire{\txi\dots}}{\tn}{\sA^+}} instead of just \im{\tx:\sA^+}.
Next we produce the environment type \im{\tnsigmaty{(\txi:\sA^+\dots)}}, from the
free source variables \im{\sxi\dots} of types \im{\sAi\dots}.
We create the environment \im{\tetwo} by creating the dependent n-tuple \im{\tnpaire{\txi\dots}};
these free variables will be replaced by values at run time.

\FigDFVs
To compute the sequence of free variables and their types, we define the metafunction
\im{\DFV{\slenv}{\se,\sB}} in \fullref[]{fig:cc:dfvs}.
Just from the syntax of terms \im{\se,\sB}, we can compute some sequence of free variables
\im{\sxin{0},\dots,\sxin{n} = \FV{(\se,\sB)}}.
However, the types of these free variables \im{\sAin{0},\dots,\sAin{n}} may contain \emph{other} free
variables, and their types may contain still others, and so on!
We must, therefore, recursively compute the a sequence of free variables and their types with respect
to an environment \im{\slenv}.
Note that because the type \im{\sB} of a term \im{\se} may contain different free variables than the
term, we must compute the sequence with respect to both a term and its type.
However, in all recursive applications of this metafunction---\eg,
\im{\DFV{\slenv}{\sAin{0},\_}}---the type of \im{\sAin{0}} must be a universe and cannot have any free
variables.

\subsection{Type Preservation}
\label{sec:cc:type-pres}
First we prove type preservation, using the same staging as in \fullref[]{sec:target}.
After we show type preservation, we show correctness of separate compilation.
In CC, the lemmas required for type preservation do most of the work to allow us to prove correctness
of separate compilation, since type checking includes reduction and thus we prove preservation of
reduction sequences.
\omitthis{
Finally we show preservation and reflection of definitional equivalence.
This relies on our model from \fullref[]{sec:target} to essential act as a \emph{back-translation},
which is commonly needed when showing equivalence preservation and reflection results.
}

We first show \emph{compositionality}.
This lemma, which establishes that translation commutes with
substitution, is the key difficulty in our proof of type preservation because closure conversion
internalizes free variables.
Whether we substitute a term for a variable before or after translation can drastically affect the
shape of closures produced by the translation.
For instance, consider the term \im{\subst{(\sfune{\sy}{\sA}{\se})}{\sepr}{\sx}}.
If we perform this substitution before translation, then we will generate an environment with the
shape \im{\tnpaire{\txi\dots,\txin{j}\dots}}, \ie, with only free variables and without \im{\tx} in
the environment.
However, if we translate the individual components and then perform the substitution, then the
environment will have the shape \im{\tnpaire{\txi\dots,\se^{\sprime+},\txin{j}\dots}}---that is,
\im{\tx} would be free when we create the environment and substitution would replace it by
\im{\se^{\sprime+}}.
We use our \(\eta\)-principle for closures to show that closures that differ in this way are still
equivalent.
\begin{lemma}[Compositionality]
  \label{lem:cc:subst}
  \im{\tr{(\subst{\seone}{\setwo}{\sx})} \equiv \subst{\tr{\seone}}{\tr{\setwo}}{\sx}}
\end{lemma}
\begin{proof}
  By induction on the typing derivation for \im{\seone}. We give the key cases.
  \begin{itemize}
    \case \rulename{Ax-Var}

    \noindent We know that \im{\seone} is some free variable \im{\sxpr}, so either
    \im{\sxpr = \sx}, hence \im{\setwo^+ \equiv \setwo^+}, or
    \im{\sxpr \neq \sx}, hence \im{\sx^{\sprime+} \equiv \sx^{\sprime+}}.

    \case \rulename{T-Code-*}

    \noindent We know that \im{\seone = \spity{\sxpr}{\sA}{\sB}}. W.l.o.g., assume \im{\sxpr \neq \sx}.
    We must show
    \im{(\spity{\sxpr}{\subst{\sA}{\setwo}{\sx}}{\subst{\sB}{\setwo}{\sx}})^+ \equiv
      \subst{(\spity{\sxpr}{\sA}{\sB})^+}{\setwo^+}{\tx}}.
    \paperonly{\vspace{-2ex}}
    \begin{align}
      & (\spity{\sxpr}{\subst{\sA}{\setwo}{\sx}}{\subst{\sB}{\setwo}{\sx}})^+ \\
      =~& \tpity{\txpr}{(\subst{\sA}{\setwo}{\sx})^+}{(\subst{\sB}{\setwo}{\sx})^+}
        \\ & \text{by definition of the translation} \nonumber \\
      =~& \tpity{\txpr}{(\subst{\sA^+}{\setwo^+}{\tx})}{(\subst{\sB^+}{\setwo^+}{\tx})}
        \\ & \text{by the inductive hypothesis for \im{\sA} and \im{\sB}}
             \nonumber \\
      =~& \subst{(\tpity{\txpr}{\sA^+}{\sB^+})}{\setwo^+}{\tx}
          \\ & \text{by definition of substitution} \nonumber \\
      =~& \subst{(\spity{\sxpr}{\sA}{\sB})^+}{\setwo^+}{\tx}
          \\ & \text{by definition of translation} \nonumber
    \end{align}

    \techrptonly{\case \rulename{Prod-$\square$}. Similar to \rulename{Prod-*}}%

    \paperonly{\vspace{-1ex}}
    \case \rulename{Lam}

    \noindent We know that \im{\seone = \sfune{\sy}{\sA}{\se}}. W.l.o.g., assume that \im{\sy \neq \sx}.
    We must show that
    \im{(\subst{(\sfune{\sy}{\sA}{\se})}{\setwo}{\sx})^+ \equiv \subst{(\sfune{\sy}{\sA}{\se})^+}{\setwo^+}{\sx}}.
    Recall that by convention we have that \im{\styjudg{\slenv}{\sfune{\sy}{\sA}{\se}}{\spity{\sy}{\sA}{\sB}}}.
    \begin{align}
      &(\subst{(\sfune{\sy}{\sA}{\se})}{\setwo}{\sx})^+ \\
      =~& (\sfune{\sy}{(\subst{\sA}{\setwo}{\sx})}{\subst{\se}{\setwo}{\sx}})^+
        \\ & \text{by substitution} \nonumber \\
      &\begin{aligned}
          \hspace{-2ex}=\tcloe{(\tnfune{&\tn:\tnsigmaty{(\txi:\sAi^+\dots)},\ty:\tlete{\tnpaire{\txi\dots}}{\tn}{(\subst{\sA}{\setwo}{\sx})^+}}{
          \\&\tlete{\tnpaire{\txi\dots}}{\tn}{(\subst{\se}{\setwo}{\sx})^+}})}{\tnpaire{\txi\dots}}
      \end{aligned}
      \\ & \text{by definition of the translation} \nonumber
    \end{align}
    where \im{\sxi:\sAi\dots {}= \DFV{\slenv}{\sfune{\sy}{(\subst{\sA}{\setwo}{\sx})}{\subst{\se}{\setwo}{\sx}}}}.
    Note that \im{\sx} is not in the sequence \im{(\sxi\dots)}.

    On the other hand, we have
    \begin{align}
      \tf=~&\subst{(\sfune{\sy}{\sA}{\se})^+}{\setwo^+}{\sx} \label{prf:tf}\\
      &\begin{aligned}
          \hspace{-2ex}=\tcloe{(\tnfune{&\tn:\tnsigmaty{(\txin{j}:\sAin{j}^+\dots)},\ty:\tlete{\tnpaire{\txin{j}\dots}}{\tn}{\sA^+}}{
          \\&\tlete{\tnpaire{\txin{j}\dots}}{\tn}{\se^+}})}{\tnpaire{\txin{j_0}\dots,\setwo^+,\txin{j_{i+1}}\dots }}
      \end{aligned}
      \\ & \text{by definition of the translation} \nonumber
    \end{align}
    where \im{\sxin{j}:\sAin{j}\dots {}= \DFV{\slenv}{\sfune{\sy}{\sA}{\se}}}.
    Note that \im{\sx} is in \im{\sxin{j}\dots};
    we can write the sequence as \im{(\sxin{j_0} \dots \sx,\sxin{j_{i+1}}\dots)}.
    Therefore, the environment we generate contains \im{\setwo^+} in position \im{j_i}.

    By \rulename{\im{\equiv}-Clo\im{_1}}, it suffices to show that

    \noindent \im{\tlete{\tnpaire{\txi\dots}}{\tnpaire{\txi\dots}}{(\subst{\se}{\setwo}{\sx})^+} \equiv
      \tappe{\tf}{\ty}}
    where \im{\tf} is the closure from \fullref[]{prf:tf}.
    \begin{align}
      \tappe{\tf}{\ty} \equiv~&
                        \tlete{\tnpaire{\txin{j_0} \dots
                               \tx,\txin{j_{i+1}\dots}}}{\tnpaire{\txin{j_0}\dots,\setwo^+,\txin{j_{i+1}}\dots}}{\se^+}
      \\& \text{by \im{\step_{\beta}} in CC-CC} \nonumber \\
      \equiv~& \subst{\se^+}{\setwo^+}{\tx}
      \\ & \text{by \len{j} applications of \im{\step_{\zeta}}, since only \im{\tx} has a value} \nonumber \\
      \equiv~& (\subst{\se}{\setwo}{\sx})^+
      \\ & \text{by the inductive hypothesis applied to the derivation for \im{\se}} \nonumber \\
      \equiv~& \tlete{\tnpaire{\txi\dots}}{\tnpaire{\txi\dots}}{(\subst{\se}{\setwo}{\sx})^+}
      \\ & \text{by \len{i} applications of \im{\step_{\zeta}}, since no variable has a value} \nonumber
    \end{align}
  \end{itemize}
\paperonly{\vspace{-4ex}\qedhere}
\paperonly{\vspace{-.5ex}}
\end{proof}

Next we show that if a source term \im{\se} takes a step, then its translation \im{\se^+} reduces in
some number of steps to a definitionally equivalent term \im{\te}.
This proof essentially follows by \fullref[]{lem:cc:subst}.
Then we show by induction on the length of the reduction sequence that the translation preserves
reduction sequences.
Note that since \fullref[]{lem:cc:subst} relies on our \(\eta\) equivalence rule for closures, we can
only show reduction up to definitional equivalence.
That is, we cannot show \im{\se^+ \stepstar {\sepr}^+}.
This is not a problem; we reason about source programs to equivalence anyway, and not up to syntactic
equality.
\begin{lemma}[Preservation of Reduction]
  \label{lem:cc:pres-red}
  If \im{\sstepjudg{\slenv}{\se}{\sepr}} then
  \im{\tstepjudg[\stepstar]{\tr{\slenv}}{\tr{\se}}{\te}} and \im{\te \equiv \se^{\sprime+}}
\end{lemma}
\begin{proof}
  By cases on \im{\sstepjudg{\slenv}{\se}{\sepr}}.
  Most cases follow easily by \fullref[]{lem:cc:subst}, since most cases of reduction are defined by substitution.
  \paperonly{We give representative cases; see our \material.}
  \begin{itemize}
    \techrptonly{
      \case \im{\sx \step_{\delta} \sepr} where \im{\sx = \sepr : \sA \in \slenv}.

      \noindent We must show that \im{\tx \stepstar \te} and \im{\se^{\sprime+} \equiv \te}.
      Let \im{\te \defeq \se^{\sprime+}}.
      It suffices to show that \im{\tx \stepstar \se^{\sprime+}}.
      By definition of the translation, we know that \im{\tx = \se^{\sprime+} : \sA^+ \in
        \tr{\slenv}} and \im{\tx \step_{\delta} \se^{\sprime+}}.

      \case \im{\slete{\sx}{\seone}{\setwo} \step_{\zeta} \subst{\setwo}{\seone}{\sx}}

      \noindent We must show that \im{(\slete{\sx}{\seone}{\setwo})^+ \stepstar \te} and
      \im{(\subst{\setwo}{\seone}{\sx})^+ \equiv \te}. Let \im{\te \defeq \subst{\setwo^+}{\seone^+}{\tx}}.
      \begin{align}
        (\slete{\sx}{\seone}{\setwo})^+ = ~&\tlete{\tx}{\seone^+}{\setwo^+} & \text{by definition of
                                                                              the translation} \\
        \step_{\zeta}~& \subst{\setwo^+}{\seone^+}{\tx} \\
        \equiv~&(\subst{\setwo}{\seone}{\sx})^+
        \\ & \text{by\fullref{lem:cc:subst}} \nonumber
      \end{align}
    }

    \case \im{\sappe{(\sfune{\sx}{\sA}{\seone})}{\setwo} \step_{\beta} \subst{\seone}{\setwo}{\sx}}

    \noindent We must show that \im{(\sappe{(\sfune{\sx}{\sA}{\seone})}{\setwo})^+ \stepstar \te} and
    \im{(\subst{\setwo}{\seone}{\sx})^+ \equiv \te}. Let \im{\te \defeq  {\subst{\seone^+}{\setwo^+}{\tx}}}.

    \noindent By definition of the translation,
    \im{(\sappe{(\sfune{\sx}{\sA}{\seone})}{\setwo})^+ = \tappe{\tf }{\setwo^+}}, where
    \begin{align}
      \tf =
        \tcloe{(\tnfune{&\tn:\tnsigmaty{(\txi:\sAi^+\dots)},\tx:\tlete{\tnpaire{\txi\dots}}{\tn}{\sA^+}}{
          \\&\tlete{\tnpaire{\txi\dots}}{\tn}{\seone^+}})}{\tnpaire{\txi\dots}}
    \end{align}

    \noindent and where \im{\sxi:\sAi\dots {}= \DFV{\slenv}{\sfune{\sx}{\sA}{\seone}}}.

    To complete the proof, observe that,
    \begin{align}
      \tappe{\tf}{\setwo^+} \step_{\beta}~
      & \tlete{\tnpaire{\txi\dots}}{\tnpaire{\txi\dots}}{\subst{\seone^+}{\setwo^+}{\tx}} \\
      {\step_{\zeta}^{\len{i}}}~& {\subst{\seone^+}{\setwo^+}{\tx}} \\
      \equiv~&(\subst{\seone}{\setwo}{\sx})^+ & \text{by \fullref[]{lem:cc:subst}}
    \end{align}
  \end{itemize}
  \paperonly{\vspace{-1ex}\qedhere}
\end{proof}

\begin{lemma}[Preservation of Reduction Sequences]
  \paperonly{\vspace{-1ex}}
  \label{lem:cc:pres-norm}
  If \im{\sstepjudg[\stepstar]{\slenv}{\se}{\sepr}} then
  \im{\tstepjudg[\stepstar]{\tr{\slenv}}{\tr{\se}}{\te}} and \im{\tequivjudg{\tr{\slenv}}{\te}{\se^{\sprime+}}}.
\end{lemma}
\techrptonly{
  \begin{proof}
    By induction on the length of the reduction sequence \im{n}.
    \begin{cases}
      \case \im{n = 0}
      Therefore \im{\sepr = \se}.

      Let \im{\te = \tr{\se}}

      By definition, \im{\tr{\se} \step^0 \tr{\se}}
      and \im{\tr{\se} \equiv \tr{\se}} by reflexivity.

      \case \im{n = i+1}
      By assumption,
      \im{\sstepjudg[\step]{\slenv}{\se}{\seone}} and
      \im{\sstepjudg[\stepstar]{\slenv}{\seone}{\sepr}}.

      It suffices to show that \im{\tr{\se} \stepstar \teone} and \im{\teone \equiv \se^{\sprime+}}.

      By \fullref[]{lem:cc:pres-red}, \im{\tr{\se} \stepstar \teone} and \im{\teone \equiv \tr{\seone}}.
      Note that our notation for translation, \im{^+}, requires that we have a typing derivation for
      \im{\seone}, thus here we rely on subject reduction of CC to know that such a derivation exists.

      It remains to be show that \im{\teone \equiv \se^{\sprime+}}.

      By the induction hypothesis, \im{\tr{\seone} \stepstar \te} and \im{\te \equiv \se^{\sprime+}}.

      Since \im{\tr{\seone} \stepstar \te}, by \rulename{\im{\equiv}}, \im{\tr{\seone} \equiv \te}.

      The goal follows by transitivity: \im{\teone \equiv \tr{\seone} \equiv \te \equiv
        \se^{\sprime+}}, therefore \im{\seone \equiv \se^{\sprime+}}.
    \end{cases}
  \end{proof}
}

We can now show \emph{coherence}, \ie, that equivalent terms are translated to equivalent terms.
As equivalence is defined primarily by \im{\stepstar}, the only interesting part of the next proof is
preserving \(\eta\) equivalence.
To show that \(\eta\) equivalence is preserved, we require our new \(\eta\) rules for closures.
\begin{lemma}[Coherence]
  \label{lem:cc:pres-equiv}
  If \im{\sequivjudg{\slenv}{\se}{\sepr}}, then
  \im{\tequivjudg{\tr{\slenv}}{\tr{\se}}{\seto{\sprime+}}}.
\end{lemma}
\begin{proof}
  By induction on the \im{\se \equiv \sepr} judgment.
  \begin{itemize}
    \techrptonly{
      \case \rulename{\im{\equiv}}

      \noindent By assumption, \im{\se \stepstar \seone} and \im{\sepr \stepstar \seone}.

      \noindent By \fullref[]{lem:cc:pres-norm}, \im{\tr{\se} \stepstar \te} and \im{\te \equiv \tr{\seone}}, and
      similarly.
      \im{\seto{\sprime+} \stepstar \tepr} and \im{\tepr \equiv \tr{\seone}}.
      The result follows by symmetry and transitivity.
    }
    \case \rulename{\im{\equiv}-\im{\eta_1}}

    \noindent By assumption, \im{\se \stepstar \sfune{\sx}{\st}{\seone}}, \im{\sepr \stepstar \setwo} and
    \im{\seone \equiv \sappe{\setwo}{\sx}}.

    \noindent Must show \im{\tr{\se} \equiv \seto{\sprime+}}.

    \noindent By \fullref[]{lem:cc:pres-norm}, \im{\tr{\se} \stepstar \te} and \im{\te \equiv
      \tr{(\sfune{\sx}{\st}{\seone})}}, and similarly \im{\seto{\sprime+} \stepstar \tepr} and \im{\tepr
      \equiv \tr{\setwo}}.

    \noindent By transitivity of \im{\equiv}, it suffices to show
    \im{\tr{(\sfune{\sx}{\st}{\seone})} \equiv \tr{\setwo}}.

    \noindent By definition of the translation,
    \begin{align}
          (\sfune{\sx}{\st}{\seone})^+ =
      \tcloe{(\tnfune{&\tn:\tnsigmaty{(\txi:\sAi^+\dots)},\tx:\tlete{\tnpaire{\txi\dots}}{\tn}{\sA^+}}{
          \nonumber \\&\tlete{\tnpaire{\txi\dots}}{\tn}{\seone^+}})}{\tnpaire{\txi\dots}} \nonumber
    \end{align}
    where \im{\sxi:\sAi\dots {}= \DFV{\slenv}{\sfune{\sx}{\st}{\seone}}}.

    \noindent By \rulename{\im{\equiv}-Clo\im{_1}} in CC-CC, it suffices to show that
    \begin{align}
      &\tlete{\tnpaire{\txi\dots}}{\tnpaire{\txi\dots}}{\seone^+} \\
      \equiv~&\seone^+ \\ & \text{by \len{i} applications of \im{\step_\zeta}} \nonumber \\
      \equiv~&\tappe{\setwo^+}{\tx}
      \\ & \text{by the inductive hypothesis applied to \im{\seone \equiv \sappe{\setwo}{\sx}}} \nonumber
    \end{align}
    \techrptonly{
      \case \rulename{\im{\equiv}-\im{\eta_2}}
      Symmetric to the previous case; requires \rulename{\im{\equiv}-\im{\eta_2}} instead of
      \rulename{\im{\equiv}-\im{\eta_1}}.
    }%
  \end{itemize}
\paperonly{\vspace{-3.9ex}}\qedhere
\end{proof}

Now we can prove type preservation.
We give the technical version of the lemma required to complete the proof, followed by the desired
statement of the theorem.
\begin{lemma}[Type Preservation (technical)]
  \label{lem:cc:type-pres}
  ~
  \begin{enumerate}
    \item If \im{\swf{\slenv}} then \im{\twf{\slenv^+}}
    \item If \im{\styjudg{\slenv}{\se}{\sA}} then \inlinemath{\ttyjudg{\slenv^+}{\se^+}{\sA^+}}
  \end{enumerate}
\end{lemma}
\begin{proof}
  Parts 1 and 2 proven simultaneously by induction on the mutually defined judgments
    \im{\swf{\slenv}} and \im{\styjudg{\slenv}{\se}{\sA}}.

  Part 1 follows easily by induction and part 2.
  We give the key cases for part 2.
  \begin{itemize}
    \case \rulename{Lam}

    \noindent We have that \im{\styjudg{\slenv}{\sfune{\sx}{\sA}{\se}}{\spity{\sx}{\sA}{\sB}}}.
    We must show that \im{\ttyjudg{\slenv^+}{(\sfune{\sx}{\sA}{\se})^+}{(\spity{\sx}{\sA}{\sB})^+}}.

    \noindent By definition of the translation, we must show that

    \noindent \im{{
          \tcloe{(\tnfune{(\begin{stackTL}\tn:\tnsigmaty{(\txi:\sAi^+\dots)},\tx:\tlete{\tnpaire{\txi\dots}}{\tn}{\sA^+})}{
              \\\tlete{\tnpaire{\txi\dots}}{\tn}{\seone^+}})}{\tnpaire{\txi\dots}}
        \end{stackTL}} : {\tpity{\tx}{\sA^+}{\sB^+}}}

    \noindent where \im{\sxi : \sAi\dots {}= \DFV{\slenv}{\sfune{\sx}{\st}{\seone}}}.

    Notice that the annotation in the term \im{\tx:\tlete{\tnpaire{\txi\dots}}{\tn}{\sA^+}}, does not match
    the annotation in the type \im{\tx : \sA^+}.
    However, by \rulename{Clo}, we can derive that the closure has type:

    \noindent \im{\tpity{(\tx}{\tlete{\tnpaire{\txi\dots}}{\tnpaire{\txi\dots}}{\sA^+})}{(\tlete{\tnpaire{\txi\dots}}{\tnpaire{\txi\dots}}{\sB^+})}},

    \noindent This is equivalent to \im{\tpity{\tx}{\sA^+}{\sB^+}} (under
    \im{\slenv^+}), since \im{(\tlete{\tnpaire{\txi\dots}}{\tnpaire{\txi\dots}}{\sA^+}) \equiv \sA^+} as
    we saw in earlier proofs.
    So, by \rulename{Clo} and \rulename{Conv}, it suffices to show that the environment and the
    code are well-typed.

    By part 1 of the induction hypothesis applied (since each of \im{\sxi:\sAi\dots} come from
    \im{\slenv}), we know the environment is well-typed:
    \im{\ttyjudg{\slenv^+}{\tnpaire{\txi\dots}}{\tnsigmaty{(\txi:\sAi^+\dots)}}}.

    Now we show that the code

    \im{(\tnfune{(\begin{stackTL}\tn:\tnsigmaty{(\txi:\sAi^+\dots)},\tx:\tlete{\tnpaire{\txi\dots}}{\tn}{\sA^+})}{
            \\\tlete{\tnpaire{\txi\dots}}{\tn}{\seone^+}})
        \end{stackTL}}

    \noindent has type \im{\tcodety{\tn,\tx}{\tlete{\tnpaire{\txi\dots}}{\tn}{\sB^+}}}. For brevity,
    we omit the duplicate type annotations on \im{\tn} and \im{\tx}.

    \noindent Observe that by the induction hypothesis applied to \im{\styjudg{\slenv}{\sA}{\sU}} and by weakening

    \im{\ttyjudg{\tn: {\tnsigmaty{(\txi:\sAi^+\dots)}}}{\tlete{\tnpaire{\txi\dots}}{\tn}{\sA^+}}{\sU^+}}.

    \noindent Hence, by \rulename{Code}, it suffices to show

    \im{\ttyjudg{\cdot,\tn,\tx}{
        \tlete{\tnpaire{\txi\dots}}{\tn}{\seone^+}}{\tlete{\tnpaire{\txi\dots}}{\tn}{\sB^+}}}

    \noindent which follows by the inductive hypothesis applied to \im{\styjudg{\slenv,\sx:\sA}{\seone}{\sB}},
    and by weakening, since \im{\sxi\dots} are the free variables of \im{\seone}, \im{\sA}, and \im{\sB}.

    \case \rulename{App}

    \noindent We have that \im{\styjudg{\slenv}{\sappe{\seone}{\setwo}}{\subst{\sB}{\setwo}{\sx}}}.
    We must show that \im{\ttyjudg{\slenv^+}{\tappe{\seone^+}{\setwo^+}}(\subst{\sB}{\setwo}{\sx})^+}.
    By \fullref[]{lem:cc:subst}, it suffices to show
    \im{\ttyjudg{\slenv^+}{\tappe{\seone^+}{\setwo^+}}\subst{\sB^+}{\setwo^+}{\tx}}, which follows by
    \rulename{App} and the inductive hypothesis applied to \im{\seone}, \im{\setwo} and \im{\sB}. \qedhere
  \end{itemize}
\end{proof}

\begin{theorem}[Type Preservation]
\paperonly{\vspace{-1ex}}
  \nonbreaking{If \im{\styjudg{\slenv}{\se}{\st}} then \im{\ttyjudg{\tr{\slenv}}{\tr{\se}}{\tr{\st}}}.}
\paperonly{\vspace{-2ex}}
\end{theorem}

\subsection{Correctness}
We prove \emph{correctness of separate compilation} and \emph{whole program correctness}.
These two theorems follow easily from \fullref[]{lem:cc:pres-norm}, but requires a little more work to
state formally.

First, we need an independent specification that relates source values to target values in CC-CC.
We do this by adding ground types, such as \im{\sBool}, to both languages and consider results related
when they are the same boolean: \im{\sTrue \approx \tTrue} and \im{\sFalse \approx \tFalse}.
It is well known how specify more sophisticated notions of observations\paperonly{, and we discuss
these in \fullref[]{sec:related}}.

Next, we define components and linking.
Components in both CC and CC-CC are well-typed open terms, \ie, \im{\styjudg{\slenv}{\se}{\sA}}.
We implement linking by substitution, and define valid closing substitutions \im{\ssubst} as follows.

\begin{inlinedisplay}
  \begin{array}{rcl}
  \wf{\slenv}{\ssubst} & \defeq & \forall \sx:\sA \in \slenv. \styjudg{\cdot}{\ssubst(\sx)}{\sA}
  \end{array}
\end{inlinedisplay}

\noindent We extend the compiler to closing substitutions \im{\ssubst^+} by point-wise application of
the translation.

Our separate compilation guarantee is that the translation of the source component \im{\se} linked
with substitution \im{\ssubst} is equivalent to first compiling \im{\se} and then linking with some
\im{\tsubst} that is definitionally equivalent to \im{\tr{\ssubst}}.
\begin{theorem}[Correctness of Separate Compilation]
  \label{thm:cc:comp-correct}
  If \im{\styjudg{\slenv}{\se}{\sA}} and \im{\sA} is a ground type,
  \im{\ssubstok{\slenv}{\ssubst}},
  \im{\ssubstok{\tr{\slenv}}{\tsubst}},
  \im{\ssubst(\se) \stepstar \sv},
  and \im{\tr{\ssubst} \equiv \tsubst}
  then
  \im{\tsubst(\tr{\se}) \stepstar \tvpr} and \im{\tr{\sv} \approx \tvpr}
\end{theorem}
\begin{proof}
  Since the translation commutes with substitution, preserves equivalence, reduction implies
  equivalence, and equivalence is transitive, the following diagram commutes.

  \begin{nop}
    \small
  \begin{tikzcd}
    \tr{(\ssubst(\se))} \arrow[r, "\equiv"] \arrow[d, "\equiv"]
    & \tsubst(\tr{\se}) \arrow[d, "\equiv"] \\
      \tr{\sv} \arrow[r, "\equiv"]
      & \tvpr
  \end{tikzcd}
  \end{nop}

  \noindent Since \im{\equiv} on ground types implies \im{\approx}, we know that \im{\sv \approx \tvpr}. \qedhere
\end{proof}

As a simple corollary, our compiler must also be whole-program correct.
If a whole-program \im{\se} evaluates to a value \im{\sv}, then the translation \im{\tr{\se}} runs to
a value equivalent to \im{\tr{\sv}}.
\begin{corollary}[Whole-Program Correctness]
  If \im{\styjudg{\cdot}{\se}{\sA}} and \im{\sA} is a ground type,
  and \im{\se \stepstar \sv}
  then \im{\tr{\se} \stepstar \tv} and \im{\tr{\sv} \approx \tv}
\end{corollary}
}

\section{Related Work and Discussion}
\label{sec:related}

\paragraph{Preserving Dependent Types}
\citet{barthe1999} study the call-by-name (CBN) CPS translation for the Calculus of Constructions
without \(\Sigma\) types.
In 2002, when attempting to extend the translation to CIC,
\citet{barthe2002} noticed that in the presence of \(\Sigma\) types,
the standard typed CPS translation fails.
In recent work, \citet{bowman2018:cps-sigma} show how to recover type preservation for both CBN and call-by-value CPS.
\citet{bowman2018:cps-sigma} add a new typing rule that keeps track of additional contextual information
while type checking continuations, similar to our \rulename{Clo} typing rule that keeps track of the
environment (via substitution) while type checking closures.
In CPS, a term \im{\te} is evaluated to values indirectly, by a continuation \im{\tnfune{\tx}{\tepr}}.
When resolving type equivalence, they end up requiring that \im{\te} is equivalent to \im{\tx}.
An essentially similar problem arises in the translation \(\Pi\) types described in \fullref[]{sec:idea}.
In closure conversion, we need show that a free variable \im{\tx} is the same as a projection from an
environment \im{\tfste{\tn}} when resolving type equivalence.

There is also work on typed compilation of restricted forms of
dependency, which avoid the central type theoretic difficulties we solve.
\citet{chen2010} develop a type-preserving compiler from Fine, an ML-like language with refinement
types, to a version of the .NET intermediate language with type-level computation.
This system lacks full spectrum types and can rely on computational irrelevance of type-level
arguments, unlike our setting as discussed in \fullref[]{sec:idea}.
\citet{shao2005} use CIC as an \emph{extrinsic type system} over an ML-like language, so that CIC
terms can be used as specifications and proofs, but restrict arbitrary terms from appearing in types.
They develop a type-preserving closure conversion translation for this language.
Their closure conversion is simpler to develop than ours, because the extrinsic type system
avoids the issues of computational relevance by disallowing terms from appear in types.
Instead, a separate type-level representation of a subset of terms can appear in types.
This can increase the burden of proving programs correct compared to an intrinsic system such as CC
because the programmer is required to duplicate programs into the type-level representation.

\vspace{-1ex}
\paragraph{Separate and Compositional Compilation}
In this work, we prove \fullref{thm:cc:comp-correct}
This is similar to the guarantees of SepCompCert~\cite{kang2016}.
We could support a \emph{compositional} compiler correctness result by developing a
relation independent of the compiler between source and target components to classify which components
are safe to link with.
There are well known techniques for developing such relations which we think will extend to
CC~\cite{perconti2014,stewart2015,neis2015:pilsner,bowman2015:niforfree,new2016}.

\vspace{-1ex}
\paragraph{Type-Preserving Compilation}
Type-preserving compilation has been widely used to rule out linking errors, and even extended to
statically rule out security attacks introduced by linking.
The seminal work by \citet{morrisett1998:ftotal} uses type-preserving compilation and give a safe
linking semantics to Typed Assembly Languages (TAL): linking any two
TAL components---regardless of
whether they were generated by a correct compiler or
hand-written---is guaranteed to be type and memory safe.
Our target language CC-CC provides similar guarantees to TAL, although it is still a high-level
language by comparison.
\section{Future Work}
\label{sec:future}

\paragraph{The Calculus of Inductive Constructions}
As future work, we plan to scale our translation to the Calculus of Inductive Constructions (CIC), the
core language of Coq.
There are two key challenges in scaling to CIC.

First, we need to scale our work to recursive functions.
Our translation should scale easily, but adding recursion to CC-CC will be challenging.
In CIC, recursive functions must always terminate to ensure consistency.
Coq enforces this with a syntactic guard condition that cannot be preserved by closure conversion, as
it relies on the structure of free variables.
Instead, we intend to investigate two alternatives to ensuring termination in CC-CC.
One is to compile the guard condition to inductive eliminators---essentially a
primitive form that folds over the tree structure of an inductive type and is terminating by
construction, which has been studied in Coq~\cite{gimenez1995}, but it is not clear how to encode this
in a typed assembly.
A more theoretically appealing technique is to design CC-CC with semantic termination concepts
such as sized types~\cite{gregoire2010}.
However, it is not clear how to compile Coq programs based on the guard condition to sized types.

Recursion also introduces an important performance consideration.
Abstract closure conversion introduces additional allocations and dereferences compared to the
existential type translation~\cite{minamide1996,morrisett1998:reccc}
To solve this, we need to adapt our definition of closures to enable environments to be separated from
the closures, but still hide their type.

The second challenge is how to address computational relevance.
In CC, it is simple to make a syntactic distinction between relevant and irrelevant
terms~\cite{barthe1999} based on the universes \im{\sstarty} and \im{\sboxty}, although we avoid doing
so for simplicity of the presentation.
Coq features an infinite hierarchy of universes, so the distinction is not easy to make.
We would need to design a source language in which computational relevance has already been made
explicit so that we can easily decide which terms to closure-convert.
Some work has been done on the design of such languages~\cite{mishra-linger2008:phd,mishra-linger2008:fossacs,miquel2001,barras2008}, but to our knowledge
none of the work has been used to encode Coq's computational relevance semantics.

\vspace{-1ex}
\paragraph{Full-Spectrum Dependently Typed Assembly Language}
Our ultimate goal is to compile to a dependently typed assembly-like language in which we can safely
link and then generate machine code.
We imagine the assembly could be targeted from many languages, such as Coq and OCaml, and type
checking in this assembly would serve to ensure safe linking.
This would require a general purpose typed assembly, with support for interoperating
between pure code and effectful code~\cite{patterson2017:linkingtypes} and different type
systems~\cite{ahmed2015:snapl}.

A minimal compiler for a functional language performs CPS translation, closure conversion,
heap allocation, and assembly-code generation.
Recent work solves CPS~\cite{bowman2018:cps-sigma} and we solve closure conversion, so two passes remain.

Heap allocation makes memory and allocation explicit, so we need new typing and equivalence rules
that can reason about memory and allocation.
This seems straightforward for CC, but we anticipate further challenges for CIC.
In CIC, we must allow cycles in the heap to support recursive functions but still ensure soundness and
termination.
As discussed earlier, ensuring terminating recursion alone will introduce new challenges, but other
techniques may help us solve the problem once the heap is explicit.
Linear types have been used to allow cycles in the heap but still guarantee strong
normalization~\cite{ahmed2007:l3}.
Unfortunately, linear types and dependent types are difficult to integrate~\cite{mcbride2016}.

The design of a dependently typed assembly language will be hard.
Assembly language will make machine-level concepts explicit, such as registers, word sizes, and
first-class control. We will need typing and equivalence rules
to reason about these machine-level concepts.
First-class control presents a particular challenge, since it is been shown inconsistent with
dependent type theory~\cite{herbelin2005}.
This is related to the problem of CPS and dependent types, so we anticipate that we can build on the work
of \citet{bowman2018:cps-sigma} to restrict control and regain consistency.
Even if control is not a problem, the consistency proof will introduce new challenges.
In this work, we develop a model of the CC-CC in CC to prove type safety and consistency.
This relies critically on \emph{compositionality}.
Assembly languages are typically not compositional, so this proof architecture may not scale to the
assembly language.
However, in recent work on interoperability between a high-level functional language and a typed assembly
language,~\citet{patterson2017:funtal} successfully defined a compositional
assembly language, and we are hopeful that we can extend this work to the dependently typed setting.

\vspace{-.5ex}
\paragraph{Secure Compilation}
Type preservation has been widely studied to statically enforce \emph{secure compilation}, \ie, to
statically guarantee that compiled code cannot be linked with components (attackers) that violate data
hiding, abstraction, and information flow security
properties~\cite{bowman2015:niforfree,ahmed2008,new2016,ahmed2011}.
These compilers typically prove that the translation preserves and reflects \emph{contextual
equivalence}, \ie, that the compiler is fully abstract.
As future work, we plan to investigate what security guarantees can be implied just from preservation
and reflection of \emph{definitional equivalence}, which we
conjecture holds of our translation.

In our model, we prove \fullref{lem:m:coherence}, \ie, that we can translate any two
\emph{definitionally equivalent} CC-CC terms into definitionally equivalent CC terms.
In our compiler, we prove \fullref{lem:cc:pres-equiv}, \ie, that we translate any two definitionally
equivalent CC terms into definitionally equivalent CC-CC terms.
These two lemmas resemble the statements of preservation and reflection, although in terms
of definitional equivalence instead of contextual equivalence.
% In fact, our model essentially corresponds to a \emph{back-translation}, a key technique used in
% proving full abstraction.
Since \fullref[]{lem:m:coherence} is not stated in terms of our compiler, we need the following
condition to complete the proof of preservation and reflection: \im{\se \equiv (\se^+)^\circ}, \ie,
that compiling to CC-CC and then translating back to CC is equivalent to the original term; we
conjecture this equivalence holds.

Definitional equivalence in dependently typed languages is sound, but not
necessarily complete, with respect to contextual equivalence, so even if the above conjecture holds,
more work remains to prove full abstraction.  Typed closure
conversion based on the existential-type encoding is well known to be
fully abstract~\cite{ahmed2008,new2016}.  We conjecture that our
translation is also fully abstract; the essential observation is that
CC-CC does not include any constructs that would allow an attacker to inspect the environment.
\vspace{-1ex}
\begin{acks}
We gratefully acknowledge the valuable feedback provided by Greg
Morrisett, Stephanie Weirich, and the anonymous reviewers.
We also give special thanks to Dan Grossman, without whom this paper would not have been.
Part of this work was done at Inria Paris in Fall 2017, while Amal Ahmed was a Visiting
Professor and William J. Bowman held an internship.
This material is based upon work supported by the National Science
Foundation under grants CCF-1422133 and CCF-1453796, and the European
Research Council under ERC Starting Grant SECOMP (715753). 
Any opinions, findings, and conclusions or recommendations expressed
in this material are those of the authors and do not necessarily
reflect the views of our funding agencies. 
\end{acks}

%%% Local Variables:
%%% mode: latex
%%% TeX-master: "paper"
%%% End:

\let\tt=\oldtt
\bibliographystyle{ACM-Reference-Format}
\bibliography{cccc}

\end{document}